\pdfoutput=1
\documentclass[aps,pra,showpacs,twoside,twocolumn,longbibliography,10pt]{revtex4-2}
\usepackage{array}[=2016-10-06]

\usepackage[colorlinks=true, citecolor=blue, urlcolor=blue, linkcolor=blue]{hyperref}

\usepackage{graphicx}        
\usepackage{bm}
\usepackage{amsmath,amsfonts,amssymb,mathtools}
\usepackage{amsthm}
\usepackage[dvipsnames]{xcolor}
\usepackage{braket}
\usepackage{float}
\usepackage{tabularx}
\usepackage{verbatim}
\usepackage{bbm}
\usepackage[normalem]{ulem}
\usepackage{enumitem}
\usepackage{fmtcount}
\usepackage{booktabs}
\usepackage[utf8]{inputenc}
\usepackage{csquotes}
\usepackage{latexsym}
\usepackage{appendix}
\usepackage{physics}
\usepackage[mathscr]{euscript}
\usepackage{esvect}
\usepackage{wrapfig}
\usepackage{dcolumn}

\newcolumntype{C}[1]{>{\centering\arraybackslash}p{#1}}

\theoremstyle{plain}

\newtheorem{theorem}{Theorem}       
\newtheorem{lemma}{Lemma}
\newtheorem{corollary}{Corollary}

\newtheorem{remark}{Remark}
\newtheorem{example}{Example}
\newtheorem{definition}{Definition}
\newtheorem{observation}{Observation}

\def\bea{\begin{eqnarray}}
\def\eea{\end{eqnarray}}
\def\ba{\begin{array}}
\def\ea{\end{array}}

\def\beq{\begin{equation}}
\def\eeq{\end{equation}}

\def\Tr{\operatorname{Tr}}

\def\T{\operatorname{T}}

\newcommand{\mc}[1]{\mathcal{#1}}

\usepackage{hyphenat}

\usepackage{orcidlink}

\begin{document}
\title{Detecting quantum properties in physical systems using proxy witnesses}

\author{Priya Ghosh\,\orcidlink{0009-0000-5908-9407}}
\affiliation{Harish-Chandra Research Institute, A CI of Homi Bhabha National Institute, Chhatnag Road, Jhunsi, Prayagraj 211019, India}

\author{Ujjwal Sen\,\orcidlink{0000-0002-0091-5847}}
\thanks{Corresponding author: ujjwal@hri.res.in}
\affiliation{Harish-Chandra Research Institute, A CI of Homi Bhabha National Institute, Chhatnag Road, Jhunsi, Prayagraj 211019, India}

\author{Siddhartha Das\,\orcidlink{0000-0002-4523-5781}}
\affiliation{q4i, Centre for Quantum Science and Technology (CQST), Center for Security, Theory and Algorithmic Research (CSTAR), International Institute of Information Technology, Hyderabad, Gachibowli, Telangana 500032, India}
\affiliation{Centre for Quantum Information \& Communication (QuIC), \'{E}cole polytechnique de Bruxelles, Universit\'{e} libre de Bruxelles, Brussels, B-1050, Belgium}
\begin{abstract}
In practice, it is quite challenging to detect a quantum property, a microscopic property, in a system. In our work, we construct general proxy witnesses of quantum properties to detect their presence in quantum systems and we do so for quantum systems which may possibly be large. In particular, we discuss proxy witnesses for quantum properties like unextendibility, quantum coherence, activation, steerability, and entanglement. We apply these proxy witnesses in some widely considered examples of many-body systems, viz., the quantum Heisenberg models and the quantum J$_1$-J$_2$ model.
\end{abstract}

\maketitle

\section{Introduction}
Properties exhibited by quantum systems that are not present in classical systems, are known to be useful resources in many information processing and computational tasks. For example, entanglement is well-known as a useful resource in quantum communication~\cite{ent-dense-coding,ent-teleportation,ent-communication-3,ent-teleportation-2}, quantum key distribution~\cite{ent-key-1,ent-key-2,DW19,DBWH21}, quantum computation~\cite{ent-computation-1,ent-computation-2,ent-computation-3,ent-computation-4}, etc. Studying quantum properties of physical systems is also of wide interest from fundamental perspectives. Measurement of macroscopic observables for a system is often easy, but detection of any microscopic observable is typically difficult.

As there are quantum properties that are not directly measurable, e.g., entanglement~\cite{entanglement}, unextendibility~\cite{W89a,doherty1,kaur1}, quantum coherence~\cite{coherence-review}, nonlocality~\cite{bell-rev}, activation~\cite{passive-3,SSC23}, etc., it is pertinent to characterize these properties and develop ``proxy" witnesses for their (indirect) detection~\cite{proxy,dowling,vedral,dagomir,lidar,gabriel}.

Unextendibility~\cite{doherty1,W89a,doherty2} is a quantum property that restricts sharing of quantum correlations between two systems in a joint state to multiple parties. The resource theory of unextendibility~\cite{kaur1,kaur2} can be seen as a relaxation of the resource theory of entanglement~\cite{relax-ent,entanglement}, and unextendibility has been studied in the context of quantum key distribution~\cite{moroder,mhyr,khatri2}, entanglement distillation~\cite{nowakowski,kaur1,berta}, quantum nonlocality~\cite{terhal,kumari}, etc. Quantum coherence~\cite{coherence} is an intriguing quantum phenomenon, presence of which is the litmus test of ``quantum" behavior and is the required resource for several quantum information processing tasks, e.g., quantum cryptography~\cite{coherence-crypto-1,coherence-crypto-2}, quantum metrology~\cite{coherence-metrology-1,coherence-metrology-2}, quantum thermodynamics~\cite{coherence-thermo-1,coherence-thermo-4, DS26}, etc.

The general difficulty in measuring
microscopic observables inspired some works where indirect means for measurement of some microscopic properties were discussed. In particular, these works proposed methods to achieve the detection of entanglement in real experiments by means of measuring physical observables, which provide an indirect way to detect entanglement~\cite{proxy,dowling,vedral,dagomir,lidar}. In this work, taking cue from previous works, we propose proxy witnesses for quantum unextendibility, quantum coherence, activation, steering, and entanglement, and present related observations for quantum systems. We first propose a few general criteria to detect quantum properties present in a system by measuring the mean values of some feasible observables. This kind of detection method is called proxy witnessing of physical properties, and the associated observables are deemed as proxy witnesses. Additionally, we show that detection of some of these quantum properties (unextendibility, coherence, activation, steering, entanglement, etc.) is possible in different physical models such as the quantum Heisenberg model, the quantum J1-J2 model through our proxy witnessing criteria.

The paper is organized as follows: In Sec.~\ref{sec2}, we discuss the preliminaries required to discuss the main results obtained in this paper. In particular, we provide definitions of extendibility, Werner and isotropic states, the quantum Heisenberg model, the quantum J$_1$-J$_2$ model, $\mathcal{PT}$- and $\mathcal{APT}$-symmetric Hamiltonians, passive and active states, quantum steering, Bell nonlocality and entanglement. In Sec.~\ref{sec3}, we mention some observations on the detection of any quantum property using proxy witnesses. In Sec.~\ref{sec4}, we provide lemmas for unextendibility witnessing and give a few applications of our unextendibility detection criterion for certain quantum many-body systems, viz. the quantum Heisenberg XXX model and the quantum J$_1$-J$_2$ model. In Sec.~\ref{sec5}, we provide a few realistic models on which quantum coherence detection through our proxy witness criteria can be achieved. In Sec.~\ref{sec6}, we find some conditions for witnessing activation experimentally and provide some applications of activation detection using proxy witnesses. In Sec.~\ref{sec7}, we discuss when steerability and entanglement in the quantum Heisenberg XXX model and in the quantum J1-J2 model can be detected through our proxy witnessing criteria. 
\textcolor{black}{In Sec.~\ref{sec9}, we discuss the experimental implementations of the observables required for proxy witnessing, including the measurement of the mean energy of a Heisenberg spin chain, and examine the effects of finite statistics, typical noise sources, and so on.} We provide our concluding remarks in Sec.~\ref{sec10}.

\section{Preliminaries}
\label{sec2}
In this section, we review some standard definitions and results from the literature that are required to derive and discuss the main results in later sections.

{Let $\mc{D}(\mc{H})$ denote the set of all density operators on $\mathcal{H}$. We denote the dimension 
of the Hilbert space $\mc{H}$ with 
$\dim(\mc{H})$. Let $\mathscr{U}(\dim(\mc{H}))$ denote the set of all unitary operators defined on a Hilbert space of dimension $\dim(\mc{H})$.  We denote the identity operator as $\mathbb{I}$ on the relevant Hilbert space and sometimes use its suffix to indicate the corresponding system. Let $\mathcal{H}_A$ and $\mathcal{H}_B$ denote the two Hilbert spaces with finite dimensions $\dim(\mc{H}_A)$ and $\dim(\mc{H}_B)$, respectively, associated with the quantum systems $A$ and $B$, so that $\mathcal{H}_{A} \otimes \mc{H}_B$ is the Hilbert space of the composite system $AB$.
 $\Gamma_{AB}$ and $F_{AB}$ respectively denote the projector on an unnormalized maximally entangled state and the swap operator:
 \begin{eqnarray}
 \Gamma_{AB} &= \sum_{i,j=0}^{d-1}\ket{ii}\bra{jj}_{AB},\label{eq:ent-op}\\
 F_{AB} &= \sum_{i,j=0}^{d-1}\ket{ij}\bra{ji}_{AB}\label{eq:swap-op},
 \end{eqnarray}
 where $d=\min\{\dim(\mc{H}_A),\dim(\mc{H}_B)\}$.}
 The swap operator $F_{AB}$ for any bipartite state also can be written as
\begin{align}
\label{eq-swap}
    F_{AB} &=  \Pi^{+}_{AB} - \Pi^{-}_{AB},
\end{align}
where $\Pi^{+}_{AB}$ and $\Pi^{-}_{AB}$ are the projections onto the symmetric and anti-symmetric subspaces of $\mathcal{H}_A \otimes \mathcal{H}_B$, and satisfies $(\Pi^{+}_{AB})^2 + (\Pi^{-}_{AB})^2 = \Pi^{+}_{AB} + \Pi^{-}_{AB}= \mathbb{I}$. Since $F_{AB}$ is the swap operator, $F_{AB}^2=\mathbb{I}_{AB}$.

\subsection{Extendibility and unextendibility}
Now we recall a quantum property called unextendibility which quantifies the unsharability of quantum correlations among multiple systems~\cite{doherty1,terhal}. In parallel, extendibility is the ability of sharing quantum correlations among multiple quantum systems. As discussed in Ref.~\cite{kaur1}, the resource theory of unextendibility can be argued to be a relaxation of the resource theory of entanglement. While the entanglement measures are monotones under local quantum operations and classical communication (LOCC channels), the measures of unextendibility are monotones under $k$-extendibility preserving channels, e.g., one-way LOCC channels, which form a strict subset of the set of LOCC channels. 

\label{subsecA}
\begin{definition}[$k$-extendible states~\cite{doherty1,W89a,doherty2}]\label{def:kex-state}
For an integer $k\geq 2$, a bipartite state $\rho_{AB}\in\mc{D}(\mc{H}_{AB})$ is  $k$-extendible with respect to a fixed subsystem $B$ if there exists a state $\Tilde{\rho}_{AB^k} \coloneqq \Tilde{\rho}_{AB_1B_2\cdots B_k}\in\mc{D}(\mc{H}_{AB_1B_2\cdots B_k})$ that satisfies
\begin{equation}\label{eq:kext}
\forall i\in \{1,\ldots,k\}:\ \rho_{AB} =\Tr_{B^k\setminus B_i}\left[\Tilde{\rho}_{AB_1\ldots B_k}\right],
\end{equation}  
where $\mc{H}_{B_i}$ is isomorphic to  $ \mc{H}_B$ for all $i\in \{1,\ldots,k\}$ and $B_1 \equiv B$. Here, $B^k = \lbrace B_1, B_2, \ldots, B_k \rbrace$.
\end{definition}

Eq.~\eqref{eq:kext} can be written in the terms of following two facts:
\begin{enumerate}
\item The state $\sigma_{AB_1B_2\cdots B_k}$ is permutation invariant with respect to the $B$ systems, in the sense that
$\forall\pi\in S_{k}$ 
\begin{equation}
\label{ext-condition-1}
\Tilde{\rho}_{AB_1B_2\cdots B_k} = 
\mathcal{W}_{B_{1}\cdots B_{k}}^{\pi}(\Tilde{\rho}_{AB_1B_2\cdots B_k}),
\end{equation}
where $\mc{W}^\pi_{B_{1}\cdots B_{k}}$ is the unitary permutation channel associated with $\pi$. Here, $S_{k}$ is the set of all permutations of the ordered set $\{1,\ldots,k\}$.
\item The state $\rho_{AB}$ is the marginal of
$\Tilde{\rho}_{AB_1 \cdots B_k}$, i.e.,
\begin{equation}
\label{ext-condition-2}
\rho_{AB} = \Tr_{B_2 \cdots  B_k}\left[\Tilde{\rho}_{AB_1\ldots B_k}\right].
\end{equation}
\end{enumerate}

\begin{definition}[k-unextendible states~\cite{doherty1,W89a,doherty2,kaur1}]
A bipartite state $\rho_{AB}$ that is not $k$-extendible w.r.t. a fixed subsystem $B$ is  called $k$-unextendible on $B$. 
\end{definition}

Let $\mathrm{EXT}_k(A;B)$ be the set of all states $\rho_{AB}\in\mathcal{D}(\mathcal{H}_{AB})$ 
that are $k$-extendible on $B$. For convenience, henceforth, extendibility and unextendibility of a bipartite state will be discussed with respect to a fixed system $B$ unless stated otherwise.  

A bipartite state is separable if and only if the state is $k$-extendible $\forall k$, $k \in \mathbbm{N}$~\cite{doherty2}. A $k$-unextendible state for any finite $k \in \mathbbm{N} (\geq 2) $ is always an entangled state. Any $k$-extendible state always remains $k$-extendible after the action of one-way local quantum operations and classical communication ($1$W-LOCC)~\cite{kaur1,nowakowski}.

In general, it appears that determining whether an arbitrary bipartite state is separable or entangled is computationally hard ($\operatorname{NP}$-hard)~\cite{gurvits,gharibian}. There, $k$-extendibility provides a way to detect entanglement. The problem of checking for $k$-extendibility in any bipartite system can be framed as a semi-definite program (SDP)~\cite{doherty1}. Semidefinite programming corresponds to the optimization of a linear function subject to linear constraints with variables lying in the cone of positive semi-definite matrices~\cite{boyd}.

\textcolor{black}{$k$-Unextendibility has emerged as a fundamental resource in quantum information science, underpinning diverse applications ranging from quantum communication and quantum cryptography to quantum channel theory, quantum error correction, and quantum many-body physics. From a resource-theoretic standpoint, the framework of $k$-unextendibility provides a semidefinite relaxation of the resource theory of entanglement, enabling detection of entanglement through semidefinite programming (SDP). Moreover, a clear hierarchy exists among $k$-extendible sets—states with smaller $k$ lie farther from separability and thus exhibit stronger entanglement than those with larger $k$. In the context of quantum communication and entanglement distillation, the resource theory of $k$-unextendibility yields tight, efficiently computable non-asymptotic upper bounds on achievable communication rates, outperforming previously known results~\cite{kaur2,kaur1}. Furthermore, in quantum cryptography, measures of unextendibility provide efficiently computable upper bounds on one-shot, one-way distillable key rates, surpassing earlier bounds based on smooth min-relative or squashed entropies~\cite{extendibility-key}. The concept of $k$-unextendibility also plays a crucial role in quantum channel theory, as a channel is anti-degradable when its Choi–Jamiołkowski state is extendible, thereby linking unextendibility to quantum capacity characterization~\cite{extendibility-channel-1,extendibility-channel-2}. In addition, optimization over $2$-PPT-extendible channels and superchannels enables efficient SDP-based lower bounds on simulation errors for approximate quantum teleportation and quantum error correction~\cite{ext-approx-teleportation}. 
}

\subsection{Werner states}
In this section, we \textcolor{black}{discuss} a special family of states called Werner states that obey certain symmetries. These states are the subject of study in different contexts of quantum information, e.g., in entanglement purification~\cite{wootters}, quantum teleportation~\cite{kim}, steering~\cite{nori}, quantum communication~\cite{vianna,dokovic,kaur1}, etc.

\label{subsecB}
Any Hermitian operator $H_{AB}$, where $\dim(\mc{H}_A)=\dim(\mc{H}_B)=d<\infty$, is called $(U \otimes U)$-invariant if for all unitary operators $U\in\mathscr{U}(d)$, we have~\cite{wer}
\begin{equation}
     H_{AB} =  (U\otimes U) H_{AB} (U^\dag\otimes U^\dag).
\end{equation}
A $(U \otimes U)$-invariant Hermitian operator $H_{AB}^{W}$ can be written as~\cite{wer}
\begin{equation}
\label{her-wer}
H_{AB}^{W} = \alpha_1^{W} \mathbb{I}_{AB} + \alpha_2^{W} F_{AB},
\end{equation}
where $\alpha_1^W$, $\alpha_2^W$ $\in \mathbbm{R}$ and $F_{AB}$ is the swap operator~\eqref{eq:swap-op}. 
 
\begin{definition}
 A Werner state $\rho^{W(p,d)}_{AB} \in\mathcal{D}(\mc{H}_{AB})$, where $\dim(\mc{H}_A)=\dim(\mc{H}_B)=d<\infty$, is one which is $(U \otimes U)$-invariant for arbitrary $U\in\mathscr{U}(d)$. A Werner state can be written as~\cite{wer}
 \begin{equation}
 \label{eqn-Werner}
\rho^{W(p,d)}_{AB} = (1-p) \frac{2}{d(d+1)} \Pi^{+}_{AB} + p \frac{2}{d(d-1)} \Pi^{-}_{AB},
\end{equation}
where $p \in [0,1]$.
\end{definition}

\begin{remark}
\label{remark-Werner}
A Werner state $\rho^{W(p,d)}_{AB}\in\mathcal{D}(\mc{H}_{AB})$ with $\dim(\mc{H}_A)=\dim(\mc{H}_B)=d$ is $k$-extendible if and only if~\cite{viola} 
\begin{equation}
p \leq \frac{1}{2}\left(\frac{d-1}{k}+1\right).   
\end{equation}
\end{remark}

For a two-qubit system $AB$, any $(U\otimes U)$-invariant Hermitian operator $H^W_{AB}$, where $\dim(\mc{H}_A)=\dim(\mc{H}_B)=2$ in Eq.~\eqref{her-wer}, can be expressed as
\begin{align}
\label{wer-sigma}
H^W_{AB} &= \left(\alpha_1^{W}+ \frac{\alpha_2^{W}}{2}\right) \mathbb{I}_{AB} +\frac{\alpha_2^{W}}{2} \sum_{i=1}^3 \sigma_i^A \otimes \sigma_i^B, 
\end{align}
because the two-qubit swap operator can be decomposed as 
$F_{AB} =  \frac{1}{2} \mathbb{I}_{AB} + \sum_{i=1}^3 \sigma_i^A \otimes \sigma_i^B$. $\{\sigma_i\}_i$, with $i\in\{1,2,3\}$, denotes the Pauli matrices.

Similarly, for a two-qutrit system $AB$, any $(U\otimes U)$-invariant Hermitian operator $H^W_{AB}$, where $\dim(\mc{H}_A)=\dim(\mc{H}_B)=3$ in Eq.~\eqref{her-wer}, can be expressed in terms of $\mathbbm{SU}(3)$ generators, as
\begin{align}
H^W_{AB} &= \left(\alpha_1^{W}+ \frac{\alpha_2^{W}}{3}\right) \mathbb{I}_{AB} +\frac{\alpha_2^{W}}{2} \sum_{j=1}^8 T_j^A \otimes T_j^B,
\end{align}
because $F_{AB} = \frac{1}{3} \mathbb{I}_{AB} +\frac{1}{2} \sum_{j=1}^8 T_j^A \otimes T_j^B $ for a two-qutrit system.
Here, $\lbrace T_j \rbrace_j$ with $j \in \lbrace 1, \cdots , 8\rbrace$ denote Gell-Mann matrices (See Appendix~\ref{gell-mat}).

\begin{remark}
\label{werner-twirling}
For every bipartite state $\rho_{AB}$ with 
$\dim(\mathcal{H}_A)=\dim(\mathcal{H}_B)=d$, there exists a Werner state
$\widetilde{\rho}^{W(p,d)}_{AB}$ defined as
\begin{align}
\label{state-changing1}
\widetilde{\rho}^{W(p,d)}_{AB}
\coloneqq
\int_{\mc{U}(d)}
(U \otimes U)\, \rho_{AB}\, (U^\dagger \otimes U^\dagger)\,
\mathrm{d}\mu(U),
\end{align}
where $\mu$ denotes the Haar measure on the unitary group
$\mc{U}(d)$. We will refer to the physical operation
$\rho_{AB} \mapsto \int_{\mc{U}(d)}
(U \otimes U)\, \rho_{AB}\, (U^\dagger \otimes U^\dagger)\,
\mathrm{d}\mu(U)$
as the \emph{Werner twirling} operation.
\end{remark}

\subsection{Isotropic states}
A Hermitian operator $H_{AB}$, where $\dim(\mc{H}_A)=\dim(\mc{H}_B)=d<\infty$, is called a $(U \otimes U^{\ast})$-invariant if~\cite{horodecki}
\begin{eqnarray}
\forall U\in\mathscr{U}(d),\  H_{AB} &=&  (U\otimes U^\ast) H_{AB} (U^\dag\otimes (U^\ast)^\dag).
\end{eqnarray}
A $(U \otimes U^{*})$-invariant Hermitian operator, $H_{AB}^{I}$, with $\dim(\mc{H}_A)=\dim(\mc{H}_B)=d$, for an arbitrary $U\in\mathscr{U}(d)$, has the form~\cite{horodecki}
\begin{equation}
\label{iso-eqn}
H_{AB}^{I} = \alpha_1^{I} \mathbb{I}_{AB} + \alpha_2^{I} \Gamma_{AB},
\end{equation}
 where $\alpha_1^I$, $\alpha_2^I$ $\in \mathbbm{R}$ and $\Gamma_{AB}$ denotes the projector on the unnormalized maximally entangled state, expressed by Eq.~\eqref{eq:ent-op}.

\begin{definition}
\label{eqn-isotropic}
An isotropic state $\rho^{I(t,d)}_{AB}\in\mathcal{D}(\mc{H}_{AB})$, where $\dim(\mc{H}_A)=\dim(\mc{H}_B)=d<\infty$, is one which is $(U\otimes U^{*})$-invariant for arbitrary $U\in\mathscr{U}(d)$. An isotropic state can be written as~\cite{horodecki}
\begin{equation}
\rho^{I(t,d)}_{AB} = t \Phi_{AB}^d + (1-t) \frac{\mathbb{I}_{AB} - \Phi_{AB}^d}{d^2 -1},
\end{equation}
where $t \in [0,1]$ and $\Phi_{AB}^d \coloneqq d \Gamma_{AB}$.  
\end{definition}

\begin{remark}
\label{remark-iso}
 An isotropic state $\rho^{I(t,d)}_{AB} \in\mathcal{D}(\mc{H}_{AB})$ with $\dim(\mc{H}_A)=\dim(\mc{H}_B)=d$  is $k$-extendible if and only if~\cite{kaur1,viola}
 \begin{equation}
  t \in \left[ 0, \frac{1}{d}\left(1+ \frac{d-1}{k}\right) \right].   
 \end{equation}
 \end{remark}

For any two-qubit system $AB$, any $(U\otimes U^{*})$-invariant Hermitian operator $H^I_{AB}$, where $\dim(\mc{H}_A)=\dim(\mc{H}_B)=2$, can be expressed as (see Appendix~\ref{appendix-iso})
\begin{equation}
H^I_{AB} =  \left(\alpha_1^{I}+ \frac{\alpha_2^{I}}{2}\right) \mathbb{I}_{AB} + \frac{\alpha_2^{I}}{2} (\sigma_1 \otimes \sigma_1 - \sigma_2 \otimes \sigma_2 + \sigma_3 \otimes \sigma_3).   
\end{equation}

\begin{remark}
\label{isotropic-twirling}
For every bipartite state $\rho_{AB}\in\mathcal{D}(\mc{H}_{AB})$, with $\dim(\mc{H}_A)=\dim(\mc{H}_B)=d$, there exists an isotropic state $\widetilde{\rho}^{I(p,d)}_{AB}$, where
\begin{align}
\label{state-changing2}
\widetilde{\rho}^{I(p,d)}_{AB}
\coloneqq
\int_{\mc{U}(d)}
(U \otimes U^\ast)\, \rho_{AB}\, (U^\dagger \otimes (U^\ast)^\dagger)\,
\mathrm{d}\mu(U),
\end{align}
where $\mu$ denotes the Haar measure on the unitary group
$\mc{U}(d)$.
The physical operations $\rho_{AB} \mapsto
\int_{\mc{U}(d)}
(U \otimes U^\ast)\, \rho_{AB}\, (U^\dagger \otimes (U^\ast)^\dagger)\,
\mathrm{d}\mu(U)$
are often referred as isotropic twirling operations.
\end{remark}

\subsection{Quantum Heisenberg and J1-J2 models}
The quantum Heisenberg model~\cite{heis1,heis2,heis3} describes a one-dimensional chain of quantum spin-1/2 particles, and we consider it to be consisting of only nearest-neighbor interactions. The Hamiltonian of the Heisenberg model of $N$ sites, following periodic boundary conditions, i.e., $\vec{\sigma}^{N+1} = \vec{\sigma}^N$, is given by
 \begin{equation}
  H_{\mathrm{H}} =  \sum_{l=1}^N \lbrack J_x \sigma_x^l \sigma_x^{l+1} + J_y 
  \sigma_y^l \sigma_y^{l+1} + J_z \sigma_z^l \sigma_z^{l+1} \rbrack,    
 \end{equation}
 where $J_\alpha \in \lbrace J_x, J_y, J_z \rbrace $ are coupling constants and
 $\lbrace \sigma_i \rbrace_i$ for $i \in \lbrace x,y,z \rbrace$ are Pauli operators. $l \in \lbrace 1, \cdots N \rbrace$ denotes site number of the chain.
$J_\alpha > 0$ represents antiferromagnetic couplings whereas $J_{\alpha} < 0$ represents ferromagnetic ones.

The Heisenberg model is called an anisotropic Heisenberg
XYZ model when 
$J_x, J_y, J_z$ are all different. In the case when $J_x = J_y \neq J_z$, the Hamiltonian describes the partially anisotropic Heisenberg XXZ model. For $J_x = J_y = J_z$, we have the isotropic Heisenberg XXX model. The model satisfying $J_x \neq J_y$, $J_z = 0$ and $J_x = J_y$, $J_z = 0$ are known as the anisotropic and isotropic XY models respectively. The XY model with an external transverse field with periodic condition is
\begin{equation}
  H_{\textnormal{XY}}^{\textnormal{field}} =  \sum_{l=1}^N \lbrack J_x \sigma_x^l \sigma_x^{l+1} + J_y \sigma_y^l \sigma_y^{l+1} \rbrack + h \sum_{l=1}^N \sigma_z^l,
 \end{equation}
where $J_x$, $J_y$ are coupling constants, $h$ is a parameter associated with the external field.
The XY model becomes an Ising model when $J_x \neq 0$, $J_y = 0$.

The ground state energy of Heisenberg XXX model in the thermodynamic limit is given by~\cite{XXXthesis,gr-energy}
\begin{align}
\label{XXX-gr}
    E^{\textnormal{XXX}}_{\textnormal{gr}} \approx - 1.77 N \abs{J}, 
\end{align}
for $ J_x = J_y = J_z = J$.
The ground state energies of the antiferromagnetic isotropic XY model and the antiferromagnetic Ising model without external field in the thermodynamic limit are~\cite{LSM}
\begin{eqnarray}
E^{\textnormal{XY}}_{\textnormal{gr}} &=& - \frac{J_x}{\pi} N\\
 \textnormal{and} \quad E^{\textnormal{Ising}}_{\textnormal{gr}} &=& - \frac{J_x}{2} N,
\end{eqnarray}
respectively, with $N$ being the number of sites of the lattice.

The Hamiltonian of the 1D J$_1$-J$_2$ model~\cite{j1-j2-3,j1-j2-4,j1-j2-1,j1-j2-2} with $N$ lattice sites, with periodic boundary conditions, is given by
\begin{align}
    H_{\textnormal{J1-J2}} = J_1 \sum_{l=1}^{N}  \vec{\sigma}^l \cdot \vec{\sigma}^{l+1} + J_2 \sum_{l=1}^{N} \vec{\sigma}^l \cdot \vec{\sigma}^{l+2},
\end{align}
where $l$ denotes site number of the lattice and $\vec{\sigma} \coloneqq \lbrace \sigma_1, \sigma_2, \sigma_3 \rbrace$.
$J_1,J_2$ denote coupling constants for nearest neighbor and next-nearest neighbor interactions of lattice respectively.

\subsection{\texorpdfstring{$\mathcal{PT}$ and $\mathcal{APT}$-symmetric Hamiltonians}{PT and APT-symmetric Hamiltonians}}
Let $\hat{P}$ and $\hat{T}$ denote the parity (P) and time reversal (T) operators, respectively. In this section, we will briefly recapitulate two classes of Hamiltonians that are categorized as $\mathcal{PT}$-symmetric and APT-symmetric~\cite{fang,PT-1,PT-qua-2,APT-1,APT-2}.

A Hamiltonian $H$ is called $\mathcal{PT}$-symmetric if it commutes with $\hat{P}\hat{T}$, i.e., $[H, \hat{P}\hat{T}] = 0$~\cite{Bender-98,Bender02,Dorey01,fang}. 
\textcolor{black}{Non-Hermitian Hamiltonians that respect parity–time ($\mathcal{PT}$) symmetry exhibit two distinct spectral regimes. In the $\mathcal{PT}$-unbroken phase, the eigenstates remain invariant under the action of the $\mathcal{PT}$ operator, and all eigenvalues are real despite the non-Hermitian nature of the Hamiltonian. In contrast, in the $\mathcal{PT}$-broken phase, the eigenstates lose $\mathcal{PT}$-invariance, and the spectrum develops complex-conjugate pairs of eigenvalues. The transition between these two regimes occurs at a critical parameter value, known as the exceptional point (EP), where at least a subset of eigenvalues and eigenvectors coalesce.}
$\mathcal{PT}$-symmetric non-Hermitian Hamiltonians display some unusual properties in a number of systems ranging from classical~\cite{PT-cla-1,PT-cla-2,PT-cla-3,PT-cla-4,PT-cla-11}
to quantum~\cite{PT-qua-1,PT-qua-2,PT-qua-3,PT-qua-4,fang}. Any $\mathcal{PT}$-symmetric Hamiltonian for a single qubit has the following form in the standard $\sigma_z$-basis~\cite{fang,PT-1,PT-qua-2}:
\begin{equation}
\label{PT-eqn}
H_{\mathcal{PT}} =  \left( \begin{array}{cc}
i\gamma & s \\
s & -i\gamma \end{array} \right),    
\end{equation}
where $s > 0$ and $a = \frac{\gamma}{s} > 0$ represent, respectively, an energy scale and a coefficient corresponding to the degree of non-hermiticity.

A Hamiltonian is called $\mathcal{APT}$-symmetric if it anti-commutes with the parity-time reversal operator $\hat{P}\hat{T}$, i.e., $\{H, \hat{P}\hat{T}\}= 0$~\cite{fang,APT-1,APT-2,non-hermitian-review}. Moreover, APT-symmetric systems exhibit some interesting effects~\cite{APT-balance,APT-4,APT-3}. The matrix representation of an APT-symmetric Hamiltonian for a single-qubit system in the standard $\sigma_z$-basis is~\cite{fang,APT-1,APT-2}
\begin{equation}
\label{APT-eqn}
H_{\mathcal{APT}} =  \left( \begin{array}{cc}
\gamma & is \\
is & -\gamma \end{array} \right),   
\end{equation}
with $s > 0$ and $a = \frac{\gamma}{s} > 0$ being, respectively, an energy scale and a coefficient corresponding to the degree of non-Hermiticity.

\textcolor{black}{
The Hamiltonian of the $\mathcal{PT}$-symmetric non-Hermitian Ising (NHTI) model~\cite{NHTI-ref-1,NHTI-ref-2,NHTI-ref-3,NHTI-ref-4,NHTI-ref-5,NHTI-ref-6} is given by
\begin{equation}
\label{eq-NHTI}
H_{\mathrm{NHTI}} = - J \sum_{l=1}^N \sigma_x^l \sigma_x^{l+1}
+ h \sum_{j=1}^N \left(\sigma_z^j + i \gamma \sigma_y^j\right),
\end{equation}
subject to periodic boundary conditions $\sigma_x^{N+1} = \sigma_x^1$.
Here, $J$, $h$, and $\gamma$ are real parameters. For $\gamma = 0$, the model reduces to the well-known Hermitian transverse-field Ising model. When $\gamma \neq 0$, the Hamiltonian becomes non-Hermitian due to the presence of an imaginary field along the $y$-direction. The NHTI model retains $\mathcal{PT}$-symmetry, yielding real eigenvalues for $|\gamma|<1$ and complex conjugate pairs for $|\gamma|>1$, with the $\mathcal{PT}$-symmetry transition occurring at $\gamma=\pm1$.
}

\subsection{Passive states and active states}
Consider a Hamiltonian $H$ associated with a quantum system $A$ with separable Hilbert space $\mathcal{H}_d$.
Let the spectral decomposition of $H$ be
\begin{equation}
\label{passive-eqn-1}
     H = \sum_{i=1}^d E_i \op{\lambda_i},
 \end{equation}
where $\lbrace \ket{\lambda_i} \rbrace_i$ is an energy eigenbasis associated with energy eigenvalues $\{E_i\}_i$, where $E_i < E_j$ if $i < j$. $\lbrace \ket{\lambda_i} \rbrace_i$ forms an orthonormal basis of $\mc{H}$, where $\mc{H}$ can be finite- or infinite-dimensional, and where the zero eigenvalues of the Hamiltonian have also been taken account.

A state $\rho \in\mc{D}(\mc{H})$ is said to be passive if and only if it can be expressed as~\cite{passive-1,passive-2}
\begin{equation}
\label{passive-eqn-2}
    \rho = \sum_{i=1}^d q_i^{\downarrow} \op{\lambda_i},
\end{equation}
with $ q_i^{\downarrow} \leq  q_j^{\downarrow}$ for all $i > j$.
All passive states are incoherent with respect to the energy basis. No work can be extracted from passive states by the action of any unitary operations~\cite{passive-1,passive-2}. See~\cite{alicki,passive-3} and references therein for connections between passivity and quantum thermodynamics.

A state is called an active state if it is not passive. All pure states except for the ground state of the Hamiltonian are active. Active states are considered as resources in the context of quantum thermodynamics~(see Ref.~\cite{passive-3}).

\textcolor{black}{\subsection{Quantum steering and Bell nonlocality}
Quantum steering is a quantum correlation in quantum information science that lies in between entanglement and Bell non-locality~\cite{steer-rev}. 
We will now discuss the steerability of a bipartite state. Let us suppose that two parties (say Alice and Bob) share a state between them, and Alice performs a local operation on Alice's side from the setting ``$x$" and obtains the outcome ``a". Thus, Alice produces a conditional state $\rho(a|x)$ on Bob's part. If the conditional state on Bob $\rho(a|x)$ can be decomposed as
\begin{equation}
\rho(a|x) = \int d\lambda p(\lambda) p(a|x,\lambda) \sigma_\lambda,
\end{equation}
then it can be surely assured that the state shared between Alice and Bob is an unsteerable state. Otherwise, the bipartite state is said to be a steerable state. Here, $\sigma_\lambda$ represents a hidden state at Bob's side prepared with probability $p(\lambda)$ and $\lambda$ represent a hidden variable correlating Alice and Bob, while $p(a|x,\lambda)$ is the probability of getting outcome $a$ for the local measurement $x$ by Alice. It is called the local hidden state model (LHS). Steerability always certifies the entanglement of any quantum state, as all separable states admit the LHS model. Any bipartite quantum state is said to be Bell non-local if the correlation between two spatially separated systems cannot be mimicked by the local hidden variable model (LHV) or violates the Bell inequality~\cite{bell}. 
Quantum steering has applications across several quantum information tasks, including one-sided device-independent quantum key distribution~\cite{steering-crypto-branciard,steering-app-crypto-2}, randomness certification~\cite{steering-app-random-certify-1,steering-app-random-certify-2}, randomness expansion~\cite{steering-app-random-expansion}, and quantum teleportation~\cite{steering-app-teleportation}, whereas Bell non-locality serves as a resource in device-independent quantum key distribution~\cite{acin-crypto-bell,bell-app-2-2009,bell-app-3-2021,bell-app-4-2021,KHD22,PGT+23,bell-app-5-2024,bell-app-6-2024} and device-independent random number generation~\cite{random-bell-1-2010,random-bell-2-2011,random-bell-3-2018,random-bell-4-2022}.}

\textcolor{black}{\subsection{Entanglement}
Since entanglement~\cite{ent-review} serves as a fundamental resource in various areas of quantum information science, such as quantum communication~\cite{ent-dense-coding,ent-teleportation,ent-communication-3,ent-teleportation-2}, quantum cryptography~\cite{ent-key-1,ent-key-2,DW19,DBWH21}, and quantum computation~\cite{ent-computation-1,ent-computation-2,ent-computation-3,ent-computation-4}, we next investigate the proxy witnessing of entanglement in multipartite quantum states in quantum many-body systems.
}

\textcolor{black}{Any multipartite state, let's say $\rho_{A_1,A_2,\ldots,A_N}$, is called fully separable if it can be decomposed as $\rho_{A_1,A_2,\ldots,A_N} = \sum_k q_k \rho_{A_1}^k \otimes \rho_{A_2}^k \otimes \cdots \rho_{A_N}^k$ with $0 \leq q_k \leq 1$ and $\sum_k q_k = 1$, where $\rho_{A_i}^k$ represents density matrices of subsystem $A_i$. In our work, we will say that any multipartite state is entangled when the state is not fully separable.  }

\textcolor{black}{
\section{General Framework for Detecting Quantum Properties via Proxy Witnessing}
\label{sec3}
In this section, we discuss the detection of general quantum properties of physical 
systems. In many physically relevant situations, quantum properties are 
not directly detectable through standard witness operators since such operators may be difficult 
to derive analytically and even more challenging to implement  experimentally.
This motivates an indirect strategy: 
instead of detecting the property 
through standard witness operators, we detect it indirectly through bounds imposed on experimentally accessible measurable quantities.
This idea leads to the notion of \emph{proxy witnessing}. The term 
``proxy'' emphasizes that the observable used in the detection procedure 
is not designed specifically to represent the property that we wish to detect. Rather, it 
acts as an indirect representative whose expectation values obey 
fundamental constraints for states that do not have the property. If any state violates 
these constraints, it must possess the property. Thus, the observable  serves as a proxy witness for certifying the property. Before presenting this framework, we first establish a useful relation for observables that remain invariant under a linear
map.
\begin{lemma}
\label{lem:sym-ham}
Consider an operator $\varrho$ over a finite-dimensional Hilbert space, an operator $A$, and a linear map $\Omega$. Let $f$ be a function such that $f(A)$ is well defined. Suppose that
\begin{equation}
\Omega(f(A))=f(A).
\end{equation}
Then
\begin{equation}
\Tr[\varrho f(A)] = \Tr[\Omega^\dagger(\varrho)f(A)].
\end{equation}
\end{lemma}
\begin{proof}
Since $\Omega(f(A))=f(A)$, we have
\begin{align}
\Tr[\varrho f(A)]
= \Tr[\varrho\,\Omega(f(A))]
= \Tr[\Omega^\dagger(\varrho)f(A)],
\end{align}
where the last equality follows from the definition of the adjoint map $\Omega^\dagger$.
\end{proof}}
\textcolor{black}{A direct consequence of the above lemma is the following corollary.
\begin{corollary}
\label{lemma-uni-op}
Consider a finite-dimensional system with a Hamiltonian $H$ such that $H$ is preserved under a bounded linear map $\Omega$, i.e., $\Omega(H)=H$.
If the system is in a state $\rho$, then the mean energy of the system ${\Tr[\rho H]}$ is equal to
 \begin{eqnarray}
 \Tr{[\rho H]} = \Tr{[\Omega^{\dagger}(\rho) H]}.
 \end{eqnarray}
 \end{corollary}}

\textcolor{black}{
\begin{observation}
\label{observation-1}
Let $P$ be a physical property of quantum states, and let $\mathscr{C}
\subseteq \mathcal{D}(\mathcal{H})$ denote the set of all states that do not
possess the property $P$. Let $\overline{\mathscr{C}}$ be its complement in
$\mathcal{D}(\mathcal{H})$, i.e., the set of states that do possess $P$.Suppose there exists a functional $f_P(\rho,\widehat{A})$, depending on the
state $\rho$ and a physical observable $\widehat{A}$, such that the range
$\left\{ f_P(\rho,\widehat{A}) : \rho \in \mathscr{C} \right\}$
characterizes all values attainable by states lacking $P$.
Then, for any $\xi \in \mathcal{D}(\mathcal{H})$, if $f_P(\xi,\widehat{A}) \notin 
\left\{ f_P(\rho,\widehat{A}) : \rho \in \mathscr{C} \right\},$
it follows that $\xi \in \overline{\mathscr{C}}$, i.e., $\xi$ possesses the
property $P$.
\end{observation}}

\textcolor{black}{Now, let's move to the construction of the proxy witness operator, which can indirectly certify the property $P$.
In a standard scenario, an operator $Z$ is called a \emph{witness operator} for the property $P$ 
(with respect to a state $\xi$) if
\begin{align*}
\Tr[\rho Z] &\ge 0, \quad \forall \rho \in \mathscr{C}, \\
\Tr[\xi Z] &< 0 .
\end{align*}
Thus, $Z$ separates at least one state possessing $P$ from all states  that do not have it.
However, constructing such a property-specific operator may be  analytically demanding and experimentally unrealistic. So, we will now try to find a particular form of witness operator, termed as ``proxy witness operator," in the proxy witnessing scenarios as follows.
\\
Suppose that $\widehat{A}$ is a physically measurable observable. 
Assume that for all states $\rho \in \mathscr{C}$,
\begin{equation}
\Tr[\rho \widehat{A}] \ge A_{\min,\mathscr{C}},
\end{equation}
where $A_{\min,\mathscr{C}} := 
\min_{\rho \in \mathscr{C}} \Tr[\rho \widehat{A}].$ The quantity $A_{\min,\mathscr{C}}$ represents the lowest expectation 
value attainable by states which do not have the property $P$.
\begin{lemma}[Proxy Witness Operator]
\label{lemma5}
Define 
\begin{equation}
Z^{\mathrm{PW}} := 
\widehat{A} - A_{\min,\mathscr{C}}\, \mathbb{I}.
\end{equation}
If a state $\xi$ satisfies
\begin{equation}
\Tr[\xi Z^{\mathrm{PW}}] < 0,
\end{equation}
then $\xi \in \overline{\mathscr{C}}$, i.e., $\xi$ possesses $P$. 
\end{lemma}
\begin{proof}
For any $\rho \in \mathscr{C}$,
\begin{equation}
\Tr[\rho Z^{\mathrm{PW}}]
= \Tr[\rho \widehat{A}] - A_{\min,\mathscr{C}}
\ge 0,
\end{equation}
by definition of $A_{\min,\mathscr{C}}$. 
Thus, if $\Tr[\xi Z^{\mathrm{PW}}] < 0$, then 
$\xi \notin \mathscr{C}$ and hence $\xi$ possesses $P$.
\end{proof}
The operator $Z^{\mathrm{PW}}$ is called a \emph{proxy witness operator} 
because it satisfies the property of being a witness operator and it is derived from a physically accessible observable rather 
than from the intrinsic mathematical definition of the property $P$.
It certifies $P$ indirectly, through experimentally measurable 
expectation values.}

\textcolor{black}{\textit{Entropy-Assisted Proxy Witnessing:} 
Suppose the von Neumann entropy satisfies $S(\rho)\ge S'$. 
Let us define
\begin{equation}
A_{\min,\mathscr{C},S'} :=
\min_{\substack{\rho \in \mathscr{C} \\ S(\rho)\ge S'}}
\Tr[\rho \widehat{A}].
\end{equation}
\begin{corollary}
If a state $\xi$ satisfies
\begin{align}
S(\xi) &\ge S', \\
\Tr[\xi \widehat{A}] &< A_{\min,\mathscr{C},S'},
\end{align}
then $\xi$ possesses $P$.
\end{corollary}
Similarly, fixing $A' = \Tr[\xi \widehat{A}]$, define
\begin{equation}
\label{eq-SDP}
S_{\max,\mathscr{C},A'} :=
\max_{\substack{\rho \in \mathscr{C} \\
\Tr[\rho \widehat{A}] = A'}}
S(\rho).
\end{equation}
If a state $\xi$ satisfies
\begin{align}
\Tr[\xi \widehat{A}] &= A', \\
S(\xi) &> S_{\max,\mathscr{C},A'},
\end{align}
then $\xi$ must belong to $\overline{\mathscr{C}}$.
\\
Thus, the proxy witnessing framework provides a general and  experimentally motivated strategy for detecting quantum properties.} 

\textcolor{black}{To illustrate the general framework developed above, we now present a simple example based on $k$-unextendibility. The purpose of this example is solely to demonstrate the construction of a proxy witness. More elaborate applications and additional examples are presented in Sec.~\ref{sec4}.
\\
For a bipartite system on \(\mathcal{H}_{A}\otimes\mathcal{H}_{B}\), let \(\mathrm{EXT}_k(A;B)\) denote the set of states that are \(k\)-extendible on \(B\). For an observable \(\widehat{A}\), define
\begin{align}
A^{\mathrm{min}}_{\mathrm{EXT}_k}
&\coloneqq
\min_{\rho\in\mathrm{EXT}_k(A;B)}
\Tr[\rho\widehat{A}],\\
A^{\mathrm{min}}_{\mathrm{EXT}_k,S'}
&\coloneqq
\min_{\substack{\rho\in\mathrm{EXT}_k(A;B)\\ S(\rho)\geq S'}}
\Tr[\rho\widehat{A}].
\end{align}
By Lemma~\ref{lemma5}, any state \(\rho\) satisfying
$\Tr[\rho\widehat{A}]<A^{\mathrm{min}}_{\mathrm{EXT}_k}$
is necessarily \(k\)-unextendible on \(B\). The associated proxy witness operator is
\begin{equation}
Z^{\mathrm{PW}}_{\mathrm{EXT}_k}
\coloneqq
\widehat{A}
-
A^{\mathrm{min}}_{\mathrm{EXT}_k}\mathbb{I}.
\end{equation}
The ordering relations among the entropy-constrained and unconstrained thresholds corresponding to \(\mathrm{EXT}_k(A;B)\) and \(\mathrm{EXT}_{k+k'}(A;B)\), for \(k\geq2\) and \(k'\in\mathbbm{N}\), are given in Appendix~\ref{inequalities}.}
\textcolor{black}{\begin{example}
Given a bipartite system with Hamiltonian $H^W_{AB}$ written in Eq.~\eqref{her-wer}, for $\dim(\mc{H}_A)=\dim(\mc{H}_B)=d$, any state $\rho_{AB}$,
will be $k$-unextendible for any $k \in \mathbbm{N}$ if the state satisfies
\begin{equation}
\label{bi-wer-eqn}
\Tr{[\rho_{AB} H^W_{AB}]} <  (\alpha_1^{W} - \alpha_2^{W} \frac{d-1}{k}),
\end{equation}
where $\alpha_1^{W}, \alpha_2^{W} \in \mathbbm{R}$ are parameters of the Hamiltonian, $H^W_{AB}$.
\end{example}
\begin{proof}
Let $\rho_{AB}$ be any $k$-extendible state belonging to the set $\mathrm{EXT}_k(A;B)$. 
Assume that the $(\textnormal{U} \otimes \textnormal{U})$-invariant Hamiltonian, $H^W_{AB}$, acts as an observable here.
We now evaluate $A^{\textnormal{min}}_{\rm EXT_k}$ for the observable $H^W_{AB}$ and call it $H^{W, \textnormal{min}}_{\rm EXT_k}$ . Mathematically, $H^{W, \textnormal{min}}_{\rm EXT_k}$ can be written as
\begin{align}
H^{W, \textnormal{min}}_{\rm EXT_k} &\coloneqq \min_{\rho_{AB} \in \mathrm{EXT}_k(A;B)} \Tr{[\rho_{AB} H^{W}_{AB}]} ,\\
 &= \min_{\rho_{AB} \in \mathrm{EXT}_k(A;B)} \Tr{[\rho_{AB} ( \alpha_1^{W} \mathbb{I}_{AB} + \alpha_2^{W} F_{AB})]},\\
 &= \min_{\rho_{AB} \in \mathrm{EXT}_k(A;B)} \left[ \alpha_1^{W} + \alpha_2^{W} \Tr{[\rho_{AB} F_{AB}]} \right],\\
 &= \min_{\rho_{AB} \in \mathrm{EXT}_k(A;B)} \left[ \alpha_1^{W} + \alpha_2^{W} (1- 2p) \right], \label{appendix-A}\\
 &= \alpha_1^{W} + \alpha_2^{W} [1 - 2 \frac{1}{2}(\frac{d-1}{k} +1)],\label{remark-2}\\
 &= \alpha_1^{W} - \alpha_2^{W} \frac{d-1}{k},
\end{align}
where Eqs.~\eqref{appendix-A},~\eqref{remark-2} are obtained from Appendix~\ref{app-wer} and Remark~\ref{remark-Werner} respectively.
This completes the proof.
\end{proof}
\begin{example}
Given a bipartite system with Hamiltonian $H^I_{AB}$, where $\dim(\mc{H}_A)=\dim(\mc{H}_B)=d$, parameterized by $\alpha_1^{I}, \alpha_2^{I} \in \mathbbm{R}$ as written in Eq.~\eqref{iso-eqn}, any bipartite state $\rho_{AB}$ will be a $k$-unextendible state if
\begin{equation}
\label{bi-iso-eqn}
\Tr{[\rho_{AB} H^I_{AB}]} <  \alpha_1^{I} 
\end{equation}
holds for any $k \in \mathbbm{N}$.
\end{example}
\begin{proof}
Here, we will evaluate $A^{\textnormal{min}}_{\rm EXT_k}$ considering $(\textnormal{U} \otimes \textnormal{U}^*)$-invariant Hamiltonian, $H^I_{AB}$, as an observable. We will term it by $H^{I, \textnormal{min}}_{\mathrm{EXT}_k(A;B)}$ and it's mathematical form will be
\begin{align}
   &H^{I, \textnormal{min}}_{\rho_{AB} \in \mathrm{EXT}_k(A;B)} \nonumber \\
 &\coloneqq \min_{\rho_{AB} \in \mathrm{EXT}_k(A;B)} \Tr{[\rho_{AB} H^{I}_{AB}]} ,\\
 &= \min_{\rho_{AB} \in \mathrm{EXT}_k(A;B)} \Tr[\rho_{AB} (\alpha_1^{I} \mathbb{I}_{AB} + \alpha_2^{I} \Gamma_{AB})],\\
 &= \min_{\rho_{AB} \in \mathrm{EXT}_k(A;B)} \left[ \alpha_1^{I} + \alpha_2^{I} \Tr{[\rho_{AB} \Gamma_{AB}]} \right],\\
 &= \min_{\rho_{AB} \in \mathrm{EXT}_k(A;B)} \left[ \alpha_1^{I} + \alpha_2^{I} (td) \right], \label{appendix-B1}\\
 &= \alpha_1^{I} \label{remark-4},
\end{align}
for any $k \in \mathbbm{N}$. We have applied Appendix~\ref{app-iso} and Remark~\ref{remark-iso} to get the Eqs.~\eqref{appendix-B1} and~\eqref{remark-4} respectively. 
This implies that any state $\rho_{AB}$ will be $k$-unextendible if the inequality~\eqref{bi-iso-eqn} holds for any $k \in \mathbbm{N}$.
\end{proof}
}

\textcolor{black}{
\subsection{General construction procedure}
\label{subsec:general-procedure}
We now summarize the common procedure used in the rest of the work to
construct proxy witnesses for different quantum properties. Let \(P\)
denote the property to be detected, and let
\(\mathscr{C}\subseteq\mathcal{D}(\mathcal{H})\) be the reference set
consisting of all states that do not possess \(P\). For a physically
accessible observable \(\widehat{A}\), define the extremal expectation
values
\begin{align}
A_{\min,\mathscr{C}}
&\coloneqq
\min_{\rho\in\mathscr{C}}
\Tr[\rho\widehat{A}],\\
A_{\max,\mathscr{C}}
&\coloneqq
\max_{\rho\in\mathscr{C}}
\Tr[\rho\widehat{A}],
\end{align}
whenever these extrema are attained. Thus, every state
\(\rho\in\mathscr{C}\) satisfies
\begin{equation}
A_{\min,\mathscr{C}}
\leq
\Tr[\rho\widehat{A}]
\leq
A_{\max,\mathscr{C}}.
\end{equation}
\\
The proxy-witness construction then consists of the following steps.
First, the reference set \(\mathscr{C}\) associated with the property
\(P\) is identified. Second, an experimentally accessible observable
\(\widehat{A}\) is chosen. Third, the minimum or maximum expectation
value of \(\widehat{A}\) over \(\mathscr{C}\) is evaluated. Finally, a
state is certified to possess \(P\) whenever its measured expectation
value lies outside the range attainable by states in \(\mathscr{C}\).
\\
More precisely, a violation of the lower bound,
\begin{equation}
\Tr[\xi\widehat{A}]
<
A_{\min,\mathscr{C}},
\end{equation}
certifies that \(\xi\notin\mathscr{C}\). In this case, the corresponding
proxy witness operator is
\begin{equation}
Z^{\mathrm{PW},-}_{\mathscr{C}}
\coloneqq
\widehat{A}
-
A_{\min,\mathscr{C}}\mathbb{I},
\end{equation}
and the condition
\(\Tr[\xi Z^{\mathrm{PW},-}_{\mathscr{C}}]<0\)
certifies the presence of \(P\).
\\
Similarly, a violation of the upper bound, $\Tr[\xi\widehat{A}]
> A_{\max,\mathscr{C}},$
also certifies that \(\xi\notin\mathscr{C}\). The corresponding proxy
witness operator may be written as
\begin{equation}
Z^{\mathrm{PW},+}_{\mathscr{C}}
\coloneqq
A_{\max,\mathscr{C}}\mathbb{I}
-
\widehat{A},
\end{equation}
for which
\(\Tr[\xi Z^{\mathrm{PW},+}_{\mathscr{C}}]<0\)
certifies \(P\).
\\
Therefore, the property-specific applications presented in the
subsequent sections differ only in the choice of the reference set
\(\mathscr{C}\), the physical observable \(\widehat{A}\), and the
evaluation of the corresponding extremal threshold. The resulting
criteria are sufficient for detecting the property \(P\); failure to
violate the relevant bound does not, in general, imply the absence of
\(P\).}

\textcolor{black}{A summary of the proxy witnessing of various physical properties, detected using our proxy witness criteria and to be discussed in detail in the subsequent sections of the paper, is presented in Table~\ref{table}.}

\begin{table*}
 \centering
\begin{tabular}{
C{3.2cm}
C{3.5cm}
C{6.8cm}
C{2.5cm}
C{1.2cm}
}
\hline
\textbf{Quantum property} &
\textbf{Reference set $\mathscr{C}$} &
\textbf{Physical model} &
\textbf{Proxy observable} &
\textbf{Section} \\
\hline
$k$-unextendibility &
$k$-extendible states &
Antiferromagnetic XXX and J$_1$--J$_2$ models &
Mean energy &
\ref{sec4} 
\\
Quantum coherence &
Incoherent states &
$\mathcal{PT}$- and $\mathcal{APT}$-symmetric Hamiltonians, XY and Ising models &
Mean energy &
\ref{sec5} 
\\
Activity &
Passive states &
Single-qudit systems &
Mean energy &
\ref{sec6} \\
Steerability &
Unsteerable states &
Antiferromagnetic XXX model &
Mean energy &
\ref{sec7} 
\\
Bell nonlocality &
Bell-local states &
Antiferromagnetic XXX model &
Mean energy &
\ref{sec7} 
\\
Entanglement &
Fully separable states &
Antiferromagnetic XXX model &
Mean energy &
\ref{sec8} \\
\hline
\hline
\end{tabular}
\caption{\textbf{Summary of the property-specific applications of the general proxy-witnessing framework.}
For each quantum property, the table outlines the reference set
$\mathscr{C}$ of states lacking the property, the physical model considered,
the experimentally accessible observable used for proxy witnessing, and the
section containing the detailed analysis.}
\label{table}
\end{table*}

\section{Proxy witnessing of \texorpdfstring{$k$-unextendibility}{k-unextendibility}}
\label{sec4}
\textcolor{black}{Following the general construction procedure introduced in Sec.~\ref{sec3}, in this section,
we discuss the property-specific realization of proxy
witnessing for \(k\)-unextendibility. 
We first formulate the $k$-unextendibility proxy witnessing problem as a semidefinite program and derive its dual
formulation. Subsequently, we apply the proxy-witness
criteria to quantum many-body systems, including the Heisenberg XXX
model and the quantum \(J_1-J_2\) model. }

\subsection{Dual problem of \texorpdfstring{$k$-unextendibility}{k-unextendibility}}
The problem, whether any bipartite system is k-extendible or not, can be cast as an SDP (semi-definite program) primal problem, a convex optimization problem. We apply it here for our $k$-unextendibility proxy witnessing problem.
To find the dual problem of the corresponding primal problem, we first have to construct a Lagrangian with the function that have to be maximized and the constraints present in the original problem. The constraints will be added to the Lagrangian using Lagrange multipliers. The variables of the original problem will be called primal variables, whereas the Lagrange multipliers are referred to as dual variables.

Let us define
\begin{align}
C_k
:=
\widehat A_{AB_1}\otimes
\mathbb I_{B_2\cdots B_k}.
\end{align}
Since, in this case, the reference set is set of $k$-extendible states, the primal SDP can then be written as
\begin{align}
A_{\mathrm{EXT}_{k}}^{\min}
=
\underset{\widetilde{\rho}_{AB^k}}{\min}
\quad&
\Tr\!\left(C_k\widetilde{\rho}_{AB^k}\right)
\label{eq:kext-primal}
\\
\text{subject to}\quad&
\widetilde{\rho}_{AB^k}\succeq0,
\nonumber\\
&
\Tr\widetilde{\rho}_{AB^k}=1,
\nonumber\\
&
W^\pi\widetilde{\rho}_{AB^k}(W^\pi)^\dagger
=
\widetilde{\rho}_{AB^k},
\qquad
\forall\,\pi\in S_k.
\nonumber
\end{align}

To derive the corresponding dual problem, we introduce a real dual variable \(\eta\) for the normalization constraint and Hermitian operators \(Z_\pi\) for the permutation-invariance constraints. The Lagrangian is
\begin{widetext}
\begin{align}
\mathcal L
\left(
\widetilde{\rho}_{AB^k},
\eta,
\{Z_\pi\}_{\pi\in S_k}
\right)
={}&
\Tr\!\left(C_k\widetilde{\rho}_{AB^k}\right)
+
\eta
\left(
1-\Tr\widetilde{\rho}_{AB^k}
\right) +
\sum_{\pi\in S_k}
\Tr\!\left[
Z_\pi
\left(
W^\pi\widetilde{\rho}_{AB^k}(W^\pi)^\dagger
-
\widetilde{\rho}_{AB^k}
\right)
\right].
\label{eq:kext-lagrangian}
\end{align}
\end{widetext}
Using the cyclicity of the trace, the Lagrangian can be rewritten as
\begin{align}
\mathcal L
=
\eta
+
\Tr\!\left[
M_k
\widetilde{\rho}_{AB^k}
\right],
\end{align}
where
\begin{align}
M_k
:=
C_k
-
\eta\mathbb I_{AB^k}
+
\sum_{\pi\in S_k}
\left[
(W^\pi)^\dagger Z_\pi W^\pi
-
Z_\pi
\right].
\label{eq:Mk}
\end{align}

The dual function is obtained by minimizing the Lagrangian over all positive semidefinite \(\widetilde{\rho}_{AB^k}\):
\begin{align}
g\left(\eta,\{Z_\pi\}\right)
:=
\inf_{\widetilde{\rho}_{AB^k}\succeq0}
\mathcal L.
\end{align}
If \(M_k\) has a negative eigenvalue, the quantity
\(\Tr(M_k\widetilde{\rho}_{AB^k})\) can be made arbitrarily negative. Therefore, the dual function is finite only when
\begin{align}
M_k\succeq0.
\end{align}
Under this condition,
\begin{align}
\inf_{\widetilde{\rho}_{AB^k}\succeq0}
\Tr\!\left(M_k\widetilde{\rho}_{AB^k}\right)
=
0,
\end{align}
and hence
\begin{align}
g\left(\eta,\{Z_\pi\}\right)=\eta.
\end{align}

Therefore, the dual SDP is
\begin{align}
\underset{\eta,\{Z_\pi\}}{\max}
\quad&
\eta
\label{eq:kext-dual}
\\
\text{subject to}\quad&
C_k
-
\eta\mathbb I_{AB^k}
+
\sum_{\pi\in S_k}
\left[
(W^\pi)^\dagger Z_\pi W^\pi
-
Z_\pi
\right]
\succeq0,
\nonumber\\
&
\eta\in\mathbb R,
\nonumber\\
&
Z_\pi=Z_\pi^\dagger,
\qquad
\forall\,\pi\in S_k.
\nonumber
\end{align}

For every primal-feasible \(\widetilde{\rho}_{AB^k}\) and every dual-feasible choice of \(\eta\) and \(\{Z_\pi\}\), weak duality gives
\begin{align}
\eta
\leq
\Tr\!\left(C_k\widetilde{\rho}_{AB^k}\right).
\end{align}
Consequently,
\begin{align}
\eta
\leq
A_{\mathrm{EXT}_{k}}^{\min}.
\end{align}
In other words, every dual-feasible value of \(\eta\) provides a lower bound on the minimum expectation value attainable by \(k\)-extendible states.

Thus, we can find $k$-unextendible states for lower-dimensional composite systems via semi-definite programming.

\textcolor{black}{\textit{A concrete SDP example for detecting \(2\)-unextendibility:}
To demonstrate the practical applicability of the SDP formulation, we consider a two-qubit system and detect \(2\)-unextendibility using the observable
$\widehat{A}_{AB}
=
-\Phi_{AB}^{2}$,
where
$\Phi_{AB}^{2}=\ket{\Phi^{+}}\bra{\Phi^{+}},$ $
\ket{\Phi^{+}}=
\frac{\ket{00}+\ket{11}}{\sqrt{2}}.$
The minimum expectation value of \(\widehat{A}_{AB}\) over the set of \(2\)-extendible states is
\begin{align}
A_{\mathrm{EXT}_{2}}^{\min}
=
\min_{\rho_{AB}\in\mathrm{EXT}_{2}(A;B)}
\Tr\!\left(
\rho_{AB}\widehat{A}_{AB}
\right).
\label{eq:ext2-min-example}
\end{align}
A state \(\rho_{AB}\) is \(2\)-extendible with respect to subsystem \(B\) if there exists a tripartite state
$\widetilde{\rho}_{AB_{1}B_{2}}$
such that
\begin{align*}
\widetilde{\rho}_{AB_{1}B_{2}}&\succeq0,
\\
\Tr\widetilde{\rho}_{AB_{1}B_{2}}&=1,
\\
W_{B_{1}B_{2}}
\widetilde{\rho}_{AB_{1}B_{2}}
W_{B_{1}B_{2}}^{\dagger}
&=
\widetilde{\rho}_{AB_{1}B_{2}},
\end{align*}
where \(W_{B_{1}B_{2}}\) is the swap operator exchanging the subsystems \(B_{1}\) and \(B_{2}\). The original state is obtained as
$\rho_{AB_{1}}=
\Tr_{B_{2}}
\widetilde{\rho}_{AB_{1}B_{2}}.$
Using
\begin{align*}
\Tr\!\left( \rho_{AB}\widehat{A}_{AB}
\right) = \Tr\!\left[ \left(
\widehat{A}_{AB_{1}}
\otimes \mathbb{I}_{B_{2}} \right)
\widetilde{\rho}_{AB_{1}B_{2}} \right],
\end{align*}
we define
$C_{2}
=
-\Phi_{AB_{1}}^{2}
\otimes
\mathbb{I}_{B_{2}}.$
The primal SDP is then given by
\begin{align}
A_{\mathrm{EXT}_{2}}^{\min}
=
\underset{\widetilde{\rho}_{AB_{1}B_{2}}}{\min}
\quad&
\Tr\!\left(
C_{2}\widetilde{\rho}_{AB_{1}B_{2}}
\right)
\label{eq:ext2-primal-example}
\\
\text{subject to}\quad&
\widetilde{\rho}_{AB_{1}B_{2}}\succeq0,
\nonumber\\
&
\Tr\widetilde{\rho}_{AB_{1}B_{2}}=1,
\nonumber\\
&
W_{B_{1}B_{2}}
\widetilde{\rho}_{AB_{1}B_{2}}
W_{B_{1}B_{2}}^{\dagger}
=
\widetilde{\rho}_{AB_{1}B_{2}}.
\nonumber
\end{align}
We solve the above SDP numerically using CVXPY with the MOSEK solver. The optimization yields
\begin{align}
A_{\mathrm{EXT}_{2}}^{\min}
=
-\frac{3}{4}.
\label{eq:ext2-optimal-value}
\end{align}
Equivalently,
\begin{align*}
\max_{\rho_{AB}\in\mathrm{EXT}_{2}(A;B)}
\Tr\!\left(
\rho_{AB}\Phi_{AB}^{2}
\right)
=
\frac{3}{4}.
\end{align*}
Therefore, the overlap of any \(2\)-extendible two-qubit state with the maximally entangled state cannot exceed \(3/4\) which matches with the known analytical condition for the $k$-extendibility of isotropic states derived in Ref.~\cite{Johnson2013}.}

\textcolor{black}{The corresponding proxy-witness operator is
\begin{align}
Z_{2}^{\mathrm{PW}}
&=
\widehat{A}_{AB}
-
A_{\mathrm{EXT}_{2}}^{\min}
\mathbb{I}_{AB}
\\
&=
-\Phi_{AB}^{2}
+
\frac{3}{4}\mathbb{I}_{AB}
\\
&=
\frac{3}{4}\mathbb{I}_{AB}
-
\Phi_{AB}^{2}.
\label{eq:ext2-proxy-witness}
\end{align}
For every \(2\)-extendible state \(\rho_{AB}\),
$\Tr\!\left(
\rho_{AB}Z_{2}^{\mathrm{PW}}
\right)
\geq0.$
Therefore, a negative expectation value of \(Z_{2}^{\mathrm{PW}}\) certifies \(2\)-unextendibility.
\\
As a test state, we consider the two-qubit isotropic state
\begin{align}
\rho^{I(t,2)}_{AB}
=
t\Phi_{AB}^{2}
+
\frac{1-t}{3}
\left(
\mathbb{I}_{AB}
-
\Phi_{AB}^{2}
\right),
\qquad
0\leq t\leq1.
\label{eq:isotropic-example}
\end{align}
Since
\begin{align}
\Tr\!\left[
\rho^{I(t,2)}_{AB}
\Phi_{AB}^{2}
\right]
=
t,
\end{align}
the expectation value of the proxy witness is
\begin{align}
\Tr\!\left[
\rho^{I(t,2)}_{AB}
Z_{2}^{\mathrm{PW}}
\right]
&=
\Tr\!\left[ \rho^{I(t,2)}_{AB}
\left( \frac{3}{4}\mathbb{I}_{AB} - \Phi_{AB}^{2}
\right) \right]
\\
&= \frac{3}{4}-t.
\label{eq:isotropic-witness-value}
\end{align}
For example, for \(t=0.80\), we obtain
\begin{align}
\Tr\!\left[ \rho^{I(0.80,2)}_{AB} Z_{2}^{\mathrm{PW}}
\right] = \frac{3}{4}-0.80
= -0.05<0.
\end{align}
Hence, the isotropic state \(\rho^{I(0.80,2)}_{AB}\) is certified to be \(2\)-unextendible.}

\textcolor{black}{Thus, this example demonstrates that the SDP formulation can be used directly to compute the minimum expectation value of a physical observable over the set of \(k\)-extendible states and thereby can do proxy witnessing for \(k\)-unextendibility.}

\subsection{Applications of proxy witnessing to
\texorpdfstring{$k$-unextendibility}{k-unextendibility}}
\textcolor{black}{
Here, we demonstrate the detection of $k$-unextendibility through
proxy witnessing in the antiferromagnetic Heisenberg XXX model and the
quantum $J_1$-$J_2$ model. In this case, the reference set is the set
of $k$-extendible states. The corresponding many-body Hamiltonian is
chosen as the observable, and the minimum mean energy attainable over
the set of $k$-extendible states provides the detection threshold.
}

\textcolor{black}{For constructing the observable required in the proxy-witness
criterion here, we consider the $(U\otimes U)$-invariant Hermitian operator
introduced in Eq.~\eqref{her-wer}. Its multipartite extension is useful
because it contains interaction terms of the form of spin-spin
couplings and therefore naturally connects the proxy witness criterion
with quantum many-body Hamiltonians.}

\begin{definition}
\label{gen-Werner}
A $(\textnormal{U} \otimes \textnormal{U})$-invariant bipartite Hermitian operator can be generalized into a $(U \otimes  \cdots \otimes U)$-invariant $2N$-partite Hermitian operator 
and it has the form
 \begin{eqnarray}
 \label{wer-gen}
 H^{W}_{\textnormal{gen}} &=& \alpha_1 \mathbb{I}_{A_1B_1......A_NB_N} + \sum _{m=1}^{N} \alpha_2^{m} ( F_{A_m B_m} \otimes \mathbb{I}_{\overline{AB}/A_m B_m} \nonumber\\ &+& F_{A_{m+1}B_m} \otimes \mathbb{I}_{\overline{AB}/A_{m+1}B_m}),
 \end{eqnarray}
where $\alpha_1$, $\alpha_2^m \in \mathbbm{R}$ 
for all $m = \lbrace 1,2,..., N \rbrace$ 
and $N \in \mathbbm{N}$. Here, $\mathbb{I}_{A_1B_1......A_NB_N}$ and $F_{A_mB_m}$ denote the identity operator on $\mathcal{H}_{A_1} \otimes \mc{H}_{B_1} \otimes \cdots \mc{H}_{A_N} \otimes \mc{H}_{B_N}$ and the swap operator between subsystems $A_m$ and $B_m$ respectively.
\end{definition}

We call a multipartite state $k$-extendible if it is $k$-extendible with respect to each subsystem of the state; otherwise, it is
called $k$-unextendible.
Let us suppose that all the $k$-extendible states of a $2N$-partite system on the Hilbert space $\mathcal{H}_{A_1} \otimes \mc{H}_{B_1} \otimes \cdots \mc{H}_{A_N} \otimes \mc{H}_{B_N}$
form the set $\mathrm{EXT}_k(A_1;B_1;...;B_N)$. 
\begin{theorem}
\label{theo-2}
The ground state of the antiferromagnetic Heisenberg XXX model is a $k$-unextendible state in the thermodynamic limit for any $k > 2$.
\end{theorem}

\begin{proof}
Substituting the form of the bipartite swap operator (from Eq.~\eqref{wer-sigma}) in Eq.~\eqref{wer-gen},
we have
\begin{align}
\label{H-XXX-1}
 H^{W}_{\textnormal{gen}} &= (\alpha_1+ \sum_{m=1}^{N} \alpha_2^m)\mathbb{I}_{A_1B_1...B_N} \nonumber\\&+ \sum_{m=1}^{N} \frac{\alpha_2^m}{2} \sum_{i=1}^{3} (\sigma^{A_m}_i \otimes \sigma^{B_m}_{i} + \sigma^{A^{m+1}}_i \otimes \sigma^{B^m}_i),
\end{align}
where the multipartite system consists of $2N$ parties. And $H^{W}_{\textnormal{gen}}$ will turn into the Hamiltonian of antiferromagnetic Heisenberg XXX model with lattice number $2N$ and coupling constant $J$, if
\begin{align}
\label{eqn-wer-gen}
 \alpha_1 &= - \sum_{m=1}^N \alpha_2^{m}, \hspace{3 mm}
 \frac{\alpha_2^m}{2} = J,
 \end{align}
holds for all $m \in \lbrace 1,2,..,N \rbrace$. It implies that $\alpha_1 = -2NJ$ and $\alpha_2^{m} = 2J$ for all $m$.

\textcolor{black}{According to the proxy-witness construction, we evaluate the minimum
mean energy of the Hamiltonian \(H^W_{\rm gen}\) over the reference set
\(\mathrm{EXT}_k(A_1;B_1;\ldots;B_N)\). For the XXX Hamiltonian this
quantity is denoted by \(E^{\min}_{\rm XXX}\). $E^{\textnormal{min}}_{\textnormal{XXX}}$ is given by}
\begin{align}
E^{\textnormal{min}}_{\textnormal{XXX}} &= \min_{\eta \in \mathrm{EXT}_k(A_1;B_1;...;B_N)} \lbrack -2NJ + \sum _{m=1}^{N} 2J (1-2p_m) \nonumber \\ &+ \sum _{m=1}^{N} 2J (1-2p'_m) \rbrack \label{E-min-XXX-eqn},
\end{align}
using Appendix~\ref{app-wer-gen}. Here $2N$, $J$ are the site number and coupling constant of the antiferromagnetic Heisenberg XXX model, respectively, while each $p_m, p'_m$ for $m \in \lbrace 1,2,\cdots N \rbrace$ represents a parameter of bipartite Werner state on $\mathbb{C}^2$.

Since the bipartite Werner state, parameterized by $p$, is a $k$-extendible state on the Hilbert space $(\mathcal{H}_{A} \otimes \mc{H}_B)$ with $\dim(\mc{H}_A)=\dim(\mc{H}_B)=d$, for $0 \leq p \leq \frac{1}{2} (\frac{d-1}{k} + 1)$, we have
\begin{align}
\label{ext-wer-gen}
     E^{\textnormal{min}}_{\textnormal{XXX}} = -2NJ (1 + \frac{2}{k}).
\end{align}
Then we can say from  Eqs.~\eqref{XXX-gr} and~\eqref{ext-wer-gen} that the ground state of the quantum antiferromagnetic Heisenberg XXX model with site number $2N$ and coupling constant $J$ will be a $k$-unextendible state for any $k \in \mathbbm{N}$ in the thermodynamic limit, if
\begin{align}
 -2\times1.77NJ &< - 2NJ( 1 + \frac{2}{k}),\\
 k &> 2.597,
\end{align}
which is always possible for any $k > 2$. Hence, the ground state of the antiferromagnetic XXX model is $k$-unextendible for any $k > 2$ in the thermodynamic limit.
\end{proof}
Note that Theorem~\ref{theo-2} only certifies that the ground state of the antiferromagnetic Heisenberg XXX model is $k$-unextendible in the thermodynamic limit for all $k>2$, it does not provide any conclusion regarding $2$-extendibility.
Similarly, we can find $k$-unextendible states of a multipartite system related to the 1D spin-1/2 J$_1$-J$_2$ model.

\textcolor{black}{Moreover, the mean energy appearing in the above proxy criterion can be reconstructed experimentally from the local spin correlators contained in the Hamiltonian. A detailed discussion of the required measurement settings and the effects of statistical and systematic imperfections is provided in Sec.~\ref{sec9}. }

\begin{theorem}
\label{theo-3}
Any $2N$-qubit state whose mean energy is lower than 
$E^{\textnormal{min}}_{\textnormal{J1-J2}} = -2N(J_1 + J_2)\left(1 + \frac{2}{k}\right)$ 
for any $k \in \mathbbm{N}$ will be $k$-unextendible when the system Hamiltonian corresponds to the antiferromagnetic J$_1$–J$_2$ model with site number $2N$ and coupling constants $J_1$ and $J_2$.
\end{theorem}

\begin{proof}
\textcolor{black}{Following the same property-specific construction used for the XXX
model, here we evaluate the minimum mean energy of the $J_1$-$J_2$ Hamiltonian  with site number $2N$ and coupling constants $J_1, J_2 $ as an observable over
$\mathrm{EXT}_k(A_1;B_1;\ldots;B_N)$. We will call it as $E^{\textnormal{min}}_{\textnormal{J1-J2}}$ which is given by}
\begin{align}
&E^{\textnormal{min}}_{\textnormal{J1-J2}} \nonumber \\ &= \min_{\eta \in \mathrm{EXT}_k(A_1;B_1;...;B_N)} [-2N(J_1+J_2) + 2J_1 \sum_{m=1}^{N} (1-2p_m) \nonumber \\&+   2J_1 \sum_{m=1}^{N} (1-2p'_m) + 2J_2 \sum_{l=1}^{N} (1-2p_l) + 2J_2 \sum_{l=1}^{N} (1-2p'_l)], \\  
&= -2N(J_1+J_2) (1 + \frac{2}{k}),
\end{align}
because a bipartite Werner state on the Hilbert space $(\mathcal{H}_{A} \otimes \mc{H}_B)$ with $\dim(\mc{H}_A)=\dim(\mc{H}_B)=d$,
parameterized by $p$, is a $k$-extendible state for any $k \in \mathbbm{N}$ when $p$ is in the range $0 \leq p \leq \frac{1}{2} (\frac{d-1}{k} + 1)$ and $d =2$ in this case. It concludes the proof. 
\end{proof}
\textcolor{black}{Thus, the proxy-witness construction provides a sufficient criterion
for detecting \(k\)-unextendibility in quantum many-body systems by
comparing the mean energy of a state with the minimum energy attainable
by \(k\)-extendible states.}

\section{Proxy witnessing of quantum coherence}
\label{sec5}
\textcolor{black}{Following the general proxy-witness construction introduced in
Sec.~\ref{sec3}, we now specialize the framework to the detection of quantum coherence. Quantum coherence is defined relative to a fixed
reference basis and constitutes a resource in quantum cryptography, quantum metrology, and quantum thermodynamics~\cite{coherence-crypto-1,coherence-crypto-2, coherence-metrology-1,coherence-metrology-2,
coherence-thermo-1,coherence-thermo-2,
coherence-thermo-3,coherence-thermo-4}.
In the present application, the reference set is the set of incoherent states, while the system Hamiltonian is chosen as the physically accessible observable. The minimum mean energy attainable by incoherent states therefore provides the property-specific threshold for certifying coherence.}
 
Any quantum state on Hilbert space with dimension $d$, say $\chi$, will be called an incoherent state with respect to a reference basis $\lbrace \ket{i} \rbrace_i$ if it is diagonal with respect to the corresponding basis, i.e., $\chi = \sum_{i=0}^d p_i \ket{i} \bra{i}$, with $p_i \geq 0$ for all $i$ and $\sum_i p_i =1$. Coherent states are those which are not incoherent with respect to the same reference basis. Now let us consider that $\mathscr{I}$ denotes a set of all incoherent states on Hilbert space with dimension $d$ with respect to a reference basis. Let us suppose that $\widehat{A}$ acts as an observable and the minimum mean value of the observable $\widehat{A}$ over the set $\mathscr{I}$ 
is referred to as $A^{\textnormal{min}}_{\textnormal{inc}}$.

\textcolor{black}{For an incoherent state
$\chi=\sum_i p_i\ket{i}\bra{i}$, one has
\begin{equation}
\Tr[\chi\widehat A]
=
\sum_{i=0}^{d-1}p_i A_{ii}.
\end{equation}
Consequently,
\begin{equation}
\label{eq:coherence-general-threshold}
A^{\textnormal{min}}_{\textnormal{inc}}
=
\min_{0\leq i\leq d-1} A_{ii}.
\end{equation}
Therefore, any state $\rho$ satisfying
$\Tr[\rho\widehat A]
<\min_i A_{ii}$
is necessarily coherent in the chosen reference basis.}

\textcolor{black}{
We now apply Eq.~\eqref{eq:coherence-general-threshold} to several
quantum systems to detect quantum coherence. Single-qubit applications involving Hermitian, $\mathcal{PT}$- and $\mathcal{APT}$-symmetric non-Hermitian Hamiltonians are
presented in Appendix~\ref{coherence-single-qubit}, 
whereas we detect coherent states in various quantum many-body systems governed by Hermitian Hamiltonians, such as isotropic and anisotropic XY models and Ising models with and without external fields, as well as in $\mathcal{PT}$-symmetric non-Hermitian Ising models. In the following,
the reference basis is the standard $\sigma_z$ basis, and the set of
incoherent $N$-qubit states is denoted by
$\mc{I}_{\textnormal{multi}}$.}

\begin{example}
\label{lem-xy-co}
Any N-qubit state with a negative mean energy will be a coherent state with respect to the standard $\sigma_z$-basis when the Hamiltonian of the isotropic $\textnormal{XY}$ model is considered to calculate the mean energy.
\end{example}

\begin{proof}
Since all the diagonal elements of the Hamiltonian of the isotropic XY model are zero for any lattice site number $N$ in the conventional $\sigma_z$-basis, the $A^{\textnormal{min}}_{\textnormal{inc}}$ over the set $\mc{I}_{\textnormal{multi}}$ considering the Hamiltonian of the isotropic $\textnormal{XY}$ model as an observable will be zero.
It completes the proof.
\end{proof}
For the same reason, the criterion written in Example~\ref{lem-xy-co} will remain same if we consider an anisotropic XY model or an Ising model with no external field as an observable to calculate the mean energy of any $N$-qubit system.

In the thermodynamic limit, the ground state energies of an antiferromagnetic isotropic XY model and an antiferromagnetic Ising model without external field are $ (- \frac{J}{\pi} N)$ and $(- \frac{J}{2} N)$ respectively, with coupling constant $J$ and lattice site number $N$.
Hence, the ground states of an antiferromagnetic isotropic XY model and an antiferromagnetic Ising model without an external field are coherent states under the standard $\sigma_z$-basis in the thermodynamic limit.

\begin{example}
Any $N$-qubit state will be a coherent state with respect to the standard $\sigma_z$-basis if the mean energy is below $(-Nh)$, when the Hamiltonian of Ising model in presence of external field is considered to be the Hamiltonian of the system. Here, the external field of the Hamiltonian of the Ising model is parameterized by $h$ and $N$ is the lattice site number. 
\end{example}

\begin{proof}
$A^{\textnormal{min}}_{\textnormal{inc}}$ over the set $\mc{I}_{\textnormal{multi}}$ is $(-Nh)$
when we choose the Hamiltonian of the Ising model with an external field as an observable.
Thus, we conclude the proof.
\end{proof}

\textcolor{black}{
\begin{example}
\label{NHTI-theo}
Consider the $\mathcal{PT}$-symmetric non-Hermitian Ising model
parameterized by $J$, $h$, and $\gamma$. Any $N$-qubit state is coherent
with respect to the standard $\sigma_z$ basis if the real part of its
mean energy, evaluated with respect to this Hamiltonian, is smaller than
$-N|h|$, where $N$ denotes the number of lattice sites.
\end{example}
\begin{remark}
Example~\ref{NHTI-theo} provides a sufficient criterion for detecting
coherence in the standard computational basis using the
$\mathcal{PT}$-symmetric non-Hermitian Ising Hamiltonian.
\end{remark}
\begin{proof}
We evaluate the minimum real part of the mean energy over all
$N$-qubit incoherent states defined with respect to the standard
computational basis, using the $\mathcal{PT}$-symmetric non-Hermitian
Ising Hamiltonian given in Eq.~\eqref{eq-NHTI}. For every incoherent
state, the expectation values of the interaction term
$\sigma_x^l\sigma_x^{l+1}$ and the non-Hermitian term
$i\gamma\sigma_y^l$ vanish in the standard $\sigma_z$ basis.
Consequently,
\begin{equation}
\min_{\chi\in \mc{I}_{\textnormal{multi}}}
\operatorname{Re}\Tr[\chi H_{\mathrm{NHTI}}]
=
-N|h|.
\end{equation}
This threshold is independent of both $J$ and $\gamma$ for any finite
$N$. Therefore, any state $\rho$ satisfying
\begin{equation}
\operatorname{Re}\Tr[\rho H_{\mathrm{NHTI}}]
<
-N|h|
\end{equation}
cannot be incoherent and must possess quantum coherence in the standard
$\sigma_z$ basis.
\end{proof}
}

\textcolor{black}{\section{Proxy witnessing of activity}
\label{sec6}
In this section, we apply the general proxy-witnessing framework to the detection of
activity. A state is passive with respect to a Hamiltonian if its mean
energy cannot be lowered by any unitary transformation acting on a
single copy of the state. Consequently, passive states have vanishing
ergotropy, whereas nonpassive states are called active states
~\cite{passive-1,passive-2,passive-3,SSC23}.
Here, the reference set is the set of passive
states. Unlike the previous applications, which involve violation of
a lower expectation-value bound, activity is detected here through violation of an upper bound on the mean energy attainable by passive
states.}

\textcolor{black}{Let $\mathscr{P}(H_A)$ denote the set of states that are passive
with respect to the Hamiltonian $H_A$.
\begin{corollary}
\label{le-gen-active} 
Let
\begin{equation}
H_A
=
\sum_{i=0}^{d-1}E_i\ket{E_i}\bra{E_i},
\qquad
E_0\leq E_1\leq\cdots\leq E_{d-1},
\end{equation}
be the Hamiltonian of a $d$-dimensional system. Any state $\rho_A$
satisfying
\begin{equation}
\Tr[\rho_A H_A]
>
\frac{1}{d}\sum_{i=0}^{d-1}E_i
\end{equation}
is active with respect to $H_A$.
\end{corollary}
\begin{proof}
Every state passive with respect to $H_A$ has the form
\begin{equation}
\rho_A^{\mathrm{pas}}
=
\sum_{i=0}^{d-1}
q_i^\downarrow\ket{E_i}\bra{E_i},
\end{equation}
where
$q_0^\downarrow
\geq q_1^\downarrow
\geq\cdots\geq q_{d-1}^\downarrow
\geq0,$ and $
\sum_{i=0}^{d-1}q_i^\downarrow=1.$
Thus,
\begin{equation}
\Tr[\rho_A^{\mathrm{pas}}H_A]
=
\sum_{i=0}^{d-1}q_i^\downarrow E_i.
\end{equation}
The maximum of this above expression over all nonincreasing probability
vectors is attained for
$q_i^\downarrow=1/d$ for every $i$. Consequently,
\begin{equation}
\max_{\rho_A^{\mathrm{pas}}\in\mathscr{P}(H_A)}
\Tr[\rho_A^{\mathrm{pas}}H_A]
=
\frac{1}{d}\sum_{i=0}^{d-1}E_i.
\end{equation}
Hence, any state whose mean energy exceeds this bound cannot belong to
$\mathscr{P}(H_A)$ and must be active.
\end{proof}
}

\textcolor{black}{For a two-level system with $E_0<E_1$, the passive-state threshold is
\begin{equation}
E_{\max,\mathscr P}
=
\frac{E_0+E_1}{2}.
\end{equation}
Since the excited energy eigenstate has mean energy $E_1$, it violates
this bound and is therefore certified to be active.
}

\section{Proxy witnessing of steerability, Bell nonlocality}
\label{sec7}
\textcolor{black}{In this section, we apply the proxy-witnessing framework to the
detection of steerability, Bell nonlocality, and entanglement in quantum many-body systems. The corresponding reference sets are the
sets of unsteerable, Bell-local, and fully separable states, respectively. For each property, the system Hamiltonian is used as the
observable, while the minimum mean energy attainable over the relevant reference set determines the detection threshold. We first consider steerability and Bell nonlocality in the antiferromagnetic Heisenberg XXX model and subsequently discuss the
corresponding entanglement criterion.}

In this work, a multipartite state is defined as \textit{unsteerable} if all its bipartite reduced states are unsteerable. If any bipartite reduced state of a multipartite state is steerable, then the multipartite state is termed steerable. Similarly, a multipartite state is considered Bell local if all its bipartite reduced states are Bell local; otherwise, the state is called Bell non-local.

\begin{lemma}
The ground state of the antiferromagnetic Heisenberg $\textnormal{XXX}$ model is steerable and Bell non-local in the thermodynamic limit.
\end{lemma}

\begin{proof}
\textcolor{black}{Let $\mathscr{S}_{\mathrm{st}}$ denote the reference set of
$2N$-partite states whose relevant bipartite reduced states are
unsteerable. For the antiferromagnetic XXX Hamiltonian, the
property-specific detection threshold is the minimum mean energy over
$\mathscr{S}_{\mathrm{st}}$, denoted by
$E^{\min}_{\mathscr{S}_{\mathrm{st}}}$.
After the simplification, $E^{\textnormal{min}}_{S_{\textnormal{st}}}$ can be expressed as
\begin{align}
&E^{\min}_{\mathscr{S}_{\mathrm{st}}} \nonumber\\
&=
\min_{\zeta\in\mathscr{S}_{\mathrm{st}}}
\left[
-2NJ
+
2J\sum_{m=1}^{N}(1-2p_m)
+
2J\sum_{m=1}^{N}(1-2p'_m)
\right]
\label{steer-eqn-1}\\
&=
-2NJ
+
4J\sum_{m=1}^{N}
\left(1-2\frac{5}{8}\right)
\label{steer-eqn-2}\\
&=
-3NJ.
\end{align}
Equation~\eqref{steer-eqn-1} follows from
Appendix~\ref{app-wer-gen}. For the two-qubit Werner states in the
parametrization of Eq.~\eqref{eqn-Werner}, the relevant
unsteerability bound implies $p_m,p'_m\leq 5/8$ for every $m$.
Since $J>0$, the minimum is attained at the largest allowed value,
$p_m=p'_m=5/8$, yielding Eq.~\eqref{steer-eqn-2}.
For a chain of $2N$ sites, the ground-state energy in the
thermodynamic limit is
\begin{equation}
E^{\mathrm{XXX}}_{\mathrm{gr}}
\approx
-2(1.77)NJ
=
-3.54NJ.
\end{equation}
Since
\begin{equation}
E^{\mathrm{XXX}}_{\mathrm{gr}}
<
E^{\min}_{\mathscr{S}_{\mathrm{st}}}
=
-3NJ,
\end{equation}
the ground state lies outside the reference set
$\mathscr{S}_{\mathrm{st}}$ and is therefore steerable.}

In a similar way, one can prove the existence of the Bell non-local nature in the ground state of the antiferromagnetic XXX model using Ref.~\cite{steer-wer-1,steer-wer-2}.
\end{proof}

Furthermore, the steerable and Bell non-local multipartite states in the quantum J$_1$-J$_2$ model can be found using the same property-specific construction.
\\
\\
\textcolor{black}{\section{Proxy witnessing of entanglement}
\label{sec8}
For entanglement detection, the reference set is the set $\mathscr{S}_{\mathrm{FS}}$ of fully separable states. The many-body Hamiltonian is chosen as the observable, and the minimum mean energy
attainable over $\mathscr{S}_{\mathrm{FS}}$ provides the corresponding entanglement threshold. Thus, any state whose mean energy lies below this threshold is necessarily entangled.}

\textcolor{black}{In Appendix~\ref{ent-XXX}, we evaluate this threshold for the antiferromagnetic Heisenberg XXX model and show that its ground-state energy lies below the minimum energy attainable by fully separable
states in the thermodynamic limit. The ground state is therefore certified to be entangled within the proxy-witnessing framework. The
same property-specific construction can be applied to the quantum $J_1$-$J_2$ model by evaluating the minimum energy over the corresponding set of fully separable states. }

\textcolor{black}{Although this result is well established in the literature, we include it here to demonstrate that entanglement detection in quantum many-body systems can be effectively achieved through our proposed proxy witness framework, providing an experimentally accessible alternative to conventional approaches. Moreover, our proof of this result is formulated entirely within the proxy witness framework, which, to the best of our knowledge, has not been previously reported. The same procedure can also be employed to detect entangled states in the quantum $J_1$–$J_2$ model.}

\textcolor{black}{\section{Experimental implementation and robustness}
\label{sec9}
The proxy-witness conditions considered in this work require only the expectation value of a suitably chosen many-body observable, rather than the reconstruction of the complete density matrix. 
For example, for the nearest-neighbor spin-$1/2$ Heisenberg model, the mean energy is determined by the one-site and two-site Pauli expectation values that appear explicitly in the Hamiltonian. In particular, the mean energy of the Heisenberg XXX model can be reconstructed from the two-body correlators
$\langle\sigma_i^x\sigma_{i+1}^x\rangle$,
$\langle\sigma_i^y\sigma_{i+1}^y\rangle$,
$\langle\sigma_i^z\sigma_{i+1}^z\rangle.$
In an experimental platform with site-resolved projective measurements, the $z$-basis correlators can be measured directly in the computational basis. The $x$- and $y$-basis correlators can be obtained by applying the corresponding single-qubit basis rotations to all sites before performing the same computational-basis measurement. Therefore, all the terms in the isotropic Heisenberg Hamiltonian can be reconstructed using only three global measurement settings, corresponding to the $x$, $y$, and $z$ bases, irrespective of the chain length.
}

\textcolor{black}{Furthermore, in an experiment, the estimated proxy observable is affected by both statistical and systematic uncertainties. Finite measurement statistics produce an uncertainty in the estimated mean energy. Suppose that the property under consideration is certified whenever
\[
\langle H\rangle<E_{\mathcal{C}}^{\min},
\]
where \(E_{\mathcal{C}}^{\min}\) is the minimum energy attainable by states that do not possess the property. If the experimentally estimated energy is \(\widehat{E}\), the inequality \(\widehat{E}<E_{\mathcal{C}}^{\min}\) alone is not sufficient for a statistically reliable certification. Instead, one may construct a one-sided confidence interval and require
\[
\widehat{E}+\delta E_{\mathrm{stat}}<E_{\mathcal{C}}^{\min},
\]
where \(\delta E_{\mathrm{stat}}\) is chosen such that the exact expectation value is no greater than \(\widehat{E}+\delta E_{\mathrm{stat}}\) at the prescribed confidence level. This uncertainty may be determined from the measured variances and the number of experimental repetitions, or bounded using standard concentration inequalities.
Measurement imperfections may additionally introduce a systematic bias in the reconstructed energy. Readout errors modify the observed outcome probabilities, while imperfect basis rotations affect the measured \(x\)- and \(y\)-correlators. These errors can be independently calibrated and incorporated into a bound \(\delta E_{\mathrm{sys}}\) on the resulting energy bias. A conservative certification condition is then
\[
\widehat{E}+\delta E_{\mathrm{stat}}+\delta E_{\mathrm{sys}}
<E_{\mathcal{C}}^{\min}.
\]
This condition ensures that the complete uncertainty interval lies below the proxy threshold and therefore prevents a false-positive certification, provided that the stated error bounds are valid.}

\textcolor{black}{Specifically, the feasibility of estimating the energy of Heisenberg-type Hamiltonians has been studied experimentally on several physical platforms. Using IBM superconducting processors, Yu \textit{et al.} measured the ground-state energies of isotropic Heisenberg and XXZ chains with up to 102 qubits~\cite{Yu2023}. 
The energy spectrum of a Heisenberg spin model has also been measured using an NMR quantum simulator~\cite{Li2011}.
In a photonic analogue simulator, Ma \textit{et al.} directly extracted the ground-state energy and pairwise Heisenberg interaction energies of a four-spin system encoded in photon polarizations~\cite{Ma2011}. 
These studies demonstrate that the energy of Heisenberg-type Hamiltonians can be reconstructed in superconducting-qubit, NMR, and photonic platforms while explicitly accounting for platform-dependent gate noise, decoherence, readout errors, imperfect quantum interference, and finite measurement statistics.}

\section{Conclusion}
\label{sec10}
Detection of quantum properties in quantum systems presents a formidable challenge. Proxy witnessing offers a promising avenue to circumvent this difficulty. By measuring certain functionals comprised of observables and states, such as mean energy, proxy witnessing allows for the practical detection of quantum properties indirectly.
In our work, we have introduced and examined criteria for the detection of various quantum properties through proxy witnessing. Specifically, we have applied these criteria to detect $k$-unextendible, coherent, active, steerable, and entangled states. Furthermore, we have extended our analysis to realistic quantum many-body physics models, including the quantum Heisenberg model, quantum  model, as well as $\mathcal{PT}$- and $\mathcal{APT}$-symmetric Hamiltonians. Our work highlights the potential of proxy witnessing as a powerful tool for exploring and harnessing quantum properties in macroscopic regimes, with important implications for both fundamental physics and practical applications.

\begin{acknowledgments}
We acknowledge partial support from the Department of Science and Technology, Government of India through the QuEST grant (grant number DST/ICPS/QUST/Theme-3/2019/120). S.D. acknowledges individual fellowship at Universit\'{e} libre de Bruxelles; this project received funding from the European Union’s Horizon 2020 research and innovation program under the Marie Skłodowska-Curie Grant Agreement No. 801505. S.D. acknowledges support from the Science and Engineering Research Board (now ANRF), Department of Science and Technology, Government of India under Grant No.~SRG/2023/000217. We acknowledge partial funding support from the National Quantum Mission, an initiative of the Department of Science and Technology,
Government of India.
\end{acknowledgments}

\section*{Appendix}
\subsection{Decomposition of bipartite maximally entangled operator on \texorpdfstring{$\mathbbm{C}^2 \otimes \mathbbm{C}^2$}{C2 ⊗ C2}}
\label{appendix-iso}
We can decompose an unnormalized bipartite maximally entangled operator ($\Gamma_{AB}$) acting on two-qubit states as
\begin{align*}
\Gamma_{AB} &=  \sum_{m,n =0}^{3} t_{mn} (\sigma_m \otimes \sigma_n),\\
 &= \frac{1}{2} (\mathbb{I} \otimes \mathbb{I} + \sigma_1 \otimes \sigma_1 - \sigma_2 \otimes \sigma_2 + \sigma_3 \otimes \sigma_3),
\end{align*}
where $\sigma_0 \coloneqq \mathbb{I}$.

\subsection{Mean value of bipartite swap operator}
\label{app-wer}
Let us consider that $\rho_{AB}$ is a bipartite state on the Hilbert space $\mc{D}(\mc{H}_{AB})$ with dimension $\dim(\mc{H}_A)=\dim(\mc{H}_B)=d$.
We can evaluate the mean value of any bipartite swap operator expressed in Eq.~\eqref{eq-swap}, acting on any bipartite state $\rho_{AB}$ and it can have the following form
\begin{align}
    \Tr{[\rho_{AB} F_{AB}]} &= \Tr{[ \Omega^{\dagger} (\rho_{AB}) F_{AB}]}, \label{mean-swap-1}\\
 &= \Tr{[ \widetilde{\rho}_{AB}^{W(p,d)} F_{AB}]}, \label{mean-swap-2}\\
  &= \Tr{\lbrack\frac{2(1-p)}{d(d+1)} \Pi^{+}_{AB} (\Pi^{+}_{AB} -\Pi^{-}_{AB})\rbrack}  \nonumber \\ &+ \Tr{\lbrack \frac{2 p}{d(d-1)} \Pi^{-}_{AB} (\Pi^{+}_{AB} -\Pi^{-}_{AB})\rbrack},\label{mean-swap-3}\\
  &=  (1-p) \frac{2}{d(d+1)} ( \frac{d(d+1)}{2} - 0)  \nonumber\\ &+  p \frac{2}{d(d-1)} ( 0 - \frac{d(d+1)}{2}) \rbrace,\\
 &= (1-2p).
\end{align}
Here, $ \widetilde{\rho}_{AB}^{W(p,d)}$ represents the Werner state on Hilbert space with dimension $d$ parameterized by $p$. 
In Eq.~\eqref{mean-swap-1}, we have used Corollary~\ref{lemma-uni-op} and the fact that $F_{AB}$ is $(U\otimes U)$-invariant Hermitian operator. We have used Remark~\ref{werner-twirling} to obtain Eq.~\eqref{mean-swap-2}.
We have substituted the form of $\widetilde{\rho}_{AB}^{W(p,d)}$ from Eq.~\eqref{eqn-Werner} in Eq.~\eqref{mean-swap-3}.

\subsection{Mean of unnormalized bipartite maximally entangled operator}
\label{app-iso}
The mean value of an unnormalized maximally entangled operator ($\Gamma_{AB}$) for any bipartite state ($\rho_{AB}$), on the Hilbert space $\mc{D}(\mc{H}_{AB})$ with dimension $\dim(\mc{H}_A)=\dim(\mc{H}_B)=d$, can be calculated as follows
\begin{align}
\Tr{[\rho_{AB} \Gamma_{AB}]} &=  \Tr{[  \Omega^{\dagger} (\rho_{AB}) \Gamma_{AB}]},\label{mean-max-1}\\ 
 &=  \Tr{[ \widetilde{\rho}_{AB}^{I(p,d)} \Gamma_{AB}]}, \label{mean-max-2}\\
&=  \Tr{[t \Phi_{AB}^d + (1-t) \frac{\mathbb{I}_{AB} - \Phi_{AB}^d}{d^2 -1} (d \Phi_{AB}^d)]},\label{mean-max-3}\\
&=   td \Tr{[(\Phi_{AB}^d)^2]}+ \frac{(1-t)d}{(d^2 -1)} (\Tr{[\Phi_{AB}^d]}\nonumber \\ &- \Tr{[(\Phi_{AB}^d)^2]}, \\
&=   td \label{mean-max-4},
\end{align}
where 
$t$ denotes a parameter of any bipartite isotropic state, $\widetilde{\rho}_{AB}^{I(p,d)}$, on the Hilbert space $\mc{D}(\mc{H}_{AB})$ with dimension $\dim(\mc{H}_A)=\dim(\mc{H}_B)=d$.
Eqs.~\eqref{mean-max-1} is achieved from Corollary~\ref{lemma-uni-op} and considering the fact that $\Gamma_{AB}$ is a $(U \otimes U^*)$-invariant Hermitian operator. We have obtained Eq.~\eqref{mean-max-2} using Remark~\ref{isotropic-twirling}.
We have put the form of $\widetilde{\rho}_{AB}^{I(p,d)}$ in Eq.~\eqref{mean-max-3} from Eq.~\eqref{eqn-isotropic}.
Note that $(\Phi_{AB}^d)^2 = \Phi_{AB}^d$, and $\Tr{[\Phi_{AB}^d]} =1$ are applied to obtain Eq.~\eqref{mean-max-4}.

\textcolor{black}{\subsection{Ordering relations among $A^{\textnormal{min}}_{\rm EXT_k}$,
$A^{\textnormal{min}}_{\rm EXT_{k},S'}$,
$A^{\textnormal{min}}_{\rm EXT_{k+k'}}$,$A^{\textnormal{min}}_{\rm EXT_{k+k'},S'}$ }
\label{inequalities}
\begin{lemma}
For any observable $\widehat{A}$, we have
\begin{align}
A^{\textnormal{min}}_{\rm EXT_k} \leq 
A^{\textnormal{min}}_{\rm EXT_{k},S'} \leq 
A^{\textnormal{min}}_{\rm EXT_{k+k'}} \leq   A^{\textnormal{min}}_{\rm EXT_{k+k'},S'},
\end{align}
for any $k\geq 2$ and $k'\in\mathbbm{N}$ (see Fig.~\ref{witness2}).
\end{lemma}
\begin{proof}
It is straightforward from the definition of $A^{\textnormal{min}}_{\rm EXT_{k},S'}$ and $ A^{\textnormal{min}}_{\rm EXT_k}$, that $ A^{\textnormal{min}}_{\rm EXT_k} \leq A^{\textnormal{min}}_{\rm EXT_{k},S'}$. Analogously, it will hold true for $\mathrm{EXT}_{k+k'}(A;B)$, i.e., $ A^{\textnormal{min}}_{\rm EXT_{k+k'}} \leq A^{\textnormal{min}}_{\rm EXT_{k+k'},S'}$.
Since there exist many states which are $k$-extendible but $(k+k')$-unextendible and the reverse is not true. Hence,
$\mathrm{EXT}_{k}(A;B)$ forms a bigger set than $\mathrm{EXT}_{k+k'}(A;B)$. Then, we can conclude that $A^{\textnormal{min}}_{\rm EXT_{k}} \leq  A^{\textnormal{min}}_{\rm EXT_{k},S'} \leq A^{\textnormal{min}}_{\rm EXT_{k+k'}}$ 
for all $\lbrace ( k \geq 2),k' \rbrace \subseteq \mathbbm{N}$. 
This concludes the proof.
\end{proof}
\begin{figure}
\includegraphics[scale =0.5]{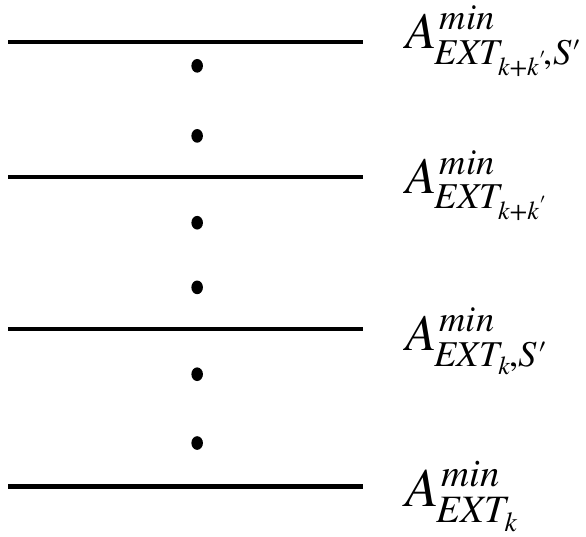}
\caption{The levels of different measurable quantities of $k$-extendibility, 
$A^{\textnormal{min}}_{\rm EXT_k}$,
$A^{\textnormal{min}}_{\rm EXT_{k},S'}$,
$A^{\textnormal{min}}_{\rm EXT_{k+k'}}$, 
and $A^{\textnormal{min}}_{\rm EXT_{k+k'},S'}$ with $\lbrace k (\geq 2), i \rbrace \subseteq \mathbbm{N}$, are shown here in increasing order from bottom to top.
}
\label{witness2}
\end{figure}}

\subsection{Evaluation of minimum mean energy over all \texorpdfstring{$k$}{k}-extendible states of any \texorpdfstring{$2N$}{2N}-partite system with Hamiltonian \texorpdfstring{$H^W_{\textnormal{gen}}$}{HWgen}}
\label{app-wer-gen}
Let us suppose that all the $k$-extendible states of a $2N$-partite system on the Hilbert space $\mathcal{H}_{A_1} \otimes \mc{H}_{B_1} \otimes \cdots \otimes \mc{H}_{A_N} \otimes \mc{H}_{B_N}$ 
form the set $\mathrm{EXT}_k(A_1;B_1;...;B_N)$ and $\eta$ denotes a state belonging to the set $\mathrm{EXT}_k(A_1;B_1;...;B_N)$.
Now, we will find the minimum mean energy over the set $\mathrm{EXT}_k(A_1;B_1;...;B_N)$ for the observable $H^W_{\textnormal{gen}}$ and it will be denoted by $E^{\textnormal{min}}_{H^{W}_{\textnormal{gen}}}$. Therefore, we have
\begin{align}
&E^{\textnormal{min}}_{H^{W}_{\textnormal{gen}}} \coloneqq \min_{\eta \in \mathrm{EXT}_k(A_1;B_1;...;B_N)} \Tr{[\eta H^{W}_\textnormal{gen}]}, \\ 
&= \min_{\eta \in \mathrm{EXT}_k(A_1;B_1;...;B_N)} 
 \Tr\lbrack \eta ( \alpha_1 \mathbb{I}_{A_1B_1......A_NB_N}  \nonumber\\ &+ \sum _{m=1}^{N} \alpha_2^{m} ( F_{A_mB_m} \otimes \mathbb{I}_{\overline{AB}/A_m B_m} \nonumber\\ &+ F_{A_{m+1}B_m} \otimes \mathbb{I}_{\overline{AB}/A_{m+1}B_m}))\rbrack , \\
&= \min_{\eta \in \mathrm{EXT}_k(A_1;B_1;...;B_N)} \lbrack \alpha_1 + \sum _{m=1}^{N} \alpha_2^{m} (1-2p_m) \nonumber \\ &+ \sum _{i=1}^{N} \alpha_2^{m} (1-2p'_m) \rbrack \label{mean-wer-gen-1}.
\end{align}

where we have taken the mathematical form of Hamiltonian $H^W_{\textnormal{gen}}$ from Eq.~\eqref{H-XXX-1}. Here, $\alpha_1$ and $\lbrace \alpha_2^m \rbrace_m$ are the parameters of the Hamiltonian $H^W_{\textnormal{gen}}$. Each and every $p_m $ and $p'_m$ for all $m$ represent a parameter of bipartite Werner state on $\mathbb{C}^2$. 
We have obtained Eq.~\eqref{mean-wer-gen-1} using Appendix~\ref{app-wer}.

\subsection{Gell-Mann matrices}
\label{gell-mat}
Eight Gell-Mann matrices are
\begin{align}
T_1 &=  \left( \begin{array}{ccc}
0 & 1 & 0 \\
1 & 0 & 0\\
0 & 0 & 0\end{array} \right),\quad  T_2 =  \left( \begin{array}{ccc}
0 & -i & 0 \\
i & 0 & 0 \\
0 & 0 & 0\end{array} \right),\nonumber \\
T_3 &= \left( \begin{array}{ccc}
1 & 0 & 0 \\
0 & -1 & 0\\
0 & 0 & 0 \end{array} \right), T_4 = \left( \begin{array}{ccc}
0 & 0 & 1 \\
0 & 0 & 0\\
1 & 0 & 0 \end{array} \right) \nonumber\\T_5 &= \left( \begin{array}{ccc}
0 & 0 & -i \\
0 & 0 & 0\\
i & 0 & 0 \end{array} \right), T_6 = \left( \begin{array}{ccc}
0 & 0 & 0 \\
0 & 0 & 1\\
0 & 1 & 0 \end{array} \right)\nonumber \\ T_7 &= \left( \begin{array}{ccc}
0 & 0 & 0 \\
0 & 0 & -i\\
0 & i & 0 \end{array} \right),T_8 = \frac{1}{\sqrt{3}} \left( \begin{array}{ccc}
1 & 0 & 0 \\
0 & 1 & 0\\
0 & 0 & -2 \end{array} \right).
\end{align}

\textcolor{black}{
\subsection{Detection of quantum coherence in single-qubit systems}
\label{coherence-single-qubit}
We can determine coherent single-qubit states with respect to the $\sigma_z$-basis through proxy witnessing. To do so, we choose a traceless Hermitian operator, $\widehat{A} \coloneqq \vec{\sigma}.\hat{n}$, as an observable. $\vec{\sigma} \coloneqq \lbrace \sigma_x, \sigma_y, \sigma_z \rbrace$ is a vector of three spin-1/2 Pauli matrices, and $\hat{n} = \lbrace \sin{\theta} \cos{\phi}, \sin{\theta}\sin{\phi}, \cos{\theta} \rbrace$ is a real three-dimensional unit vector where $\theta \in \lbrack 0, \pi \rbrack$ and $\phi \in \lbrack 0 , 2 \pi \rbrack$ represent zenith and azimuthal angles of a spherical coordinate system, respectively. 
Let us assume that all the incoherent single-qubit states of a system form a set that is denoted by $\mc{I}_{\textnormal{sq}}$. 
Let any single-qubit state be denoted by $\tau$.
\begin{example}
\label{qubit-co}
Any single-qubit state $\tau$ will be called a coherent state with respect to the standard $\sigma_z$-basis if we have $\Tr{[\tau \widehat{A}]} < - \cos \theta$ where $\theta$ denotes a parameter of the observable $\widehat{A} \coloneqq \vec{\sigma}.\hat{n}$.
\end{example}
\begin{proof}
We will calculate $A^{\textnormal{min}}_{\textnormal{inc}}$ for the observable $\vec{\sigma}.\hat{n}$ over the set $\mc{I}_\textnormal{sq}$ 
and it will be
\begin{eqnarray}
A^{\textnormal{min}}_{\textnormal{inc}} &\coloneqq& \min_{\tau \in \mc{I}_\textnormal{sq}} \Tr{[\tau \widehat{A}]},\\
 &=& \min_{\tau \in \mc{I}_\textnormal{sq}} \Tr{[\left(\sum_{i=0}^1 p_i \ket{i} \bra{i}\right) \widehat{A}]},\\
 &=& \min_{\tau \in \mc{I}_\textnormal{sq}}  \left[(p_0 - p_1) \cos{\theta}\right],\\
 &=& - \cos{\theta}.
\end{eqnarray}
This concludes the proof.
\end{proof}
It is a sufficient criterion, but not necessary. Using the criterion described above, we can determine if any single-qubit state is coherent or not with respect to the standard $\sigma_z$-basis, for example, the state $\ket{+}\bra{+}$. Here, $\ket{+}$ is an eigenstate of $\sigma_x$ with eigenvalue $1$.
The mean value of the observable $\widehat{A} \coloneqq \vec{\sigma}.\hat{n}$ for the state $\ket{+} \bra{+}$ is $\Tr{[\tau \widehat{A}]} =  \cos{\phi} \sin{\theta}$. Therefore, the state $\ket{+} \bra{+}$ will be coherent under standard $\sigma_z$-basis if $(\cos{\phi} \tan{\theta}) < -1$ holds.
Now, if we consider $\widehat{A} = - \sigma_x$, the criterion given in Example~\ref{qubit-co} satisfies. Hence, the state $\ket{+} \bra{+}$ is a coherent state with respect to the $\lbrace \ket{0}, \ket{1} \rbrace$ basis.}

\textcolor{black}{Any general single-qubit state $\tau$ can be represented as
$\tau = \frac{1}{2} (\mathbb{I} + \vec{v} \cdot \vec{\sigma})$, where $\vec{v} \coloneqq \lbrace v_x, v_y, v_z \rbrace$ is a real vector with $v^2_x + v^2_y + v^2_z \leq 1$ and $\vec{\sigma} \coloneqq \lbrace \sigma_x, \sigma_y, \sigma_z \rbrace$ is is a vector of three spin-1/2 Pauli matrices. We can witness the quantum coherence for a single qubit state via $\mathcal{PT}$- and $\mathcal{APT}$-symmetric Hamiltonians as well.}

\textcolor{black}{
\begin{example}
Any single-qubit state $\tau$ parameterized by $v_x, v_y,$ and $v_z$, will be a coherent state with respect to the standard $\sigma_z$-basis if one of the following inequalities
\begin{align}
    v_x s &< -i \gamma (1 + v_z) \label{PT-co-eqn},\\
    v_x s &< i \gamma (1 + v_z) \label{APT-co-eqn},
\end{align}
hold where the $s$ and $\gamma$ represent the parameters of $\mathcal{PT}$- and $\mathcal{APT}$-symmetric Hamiltonians.
\end{example}}

\textcolor{black}{
\begin{proof}
Considering $\mathcal{PT}$-symmetric Hamiltonian as an observable, the mean energy of any single-qubit state, $\tau$, will be
\begin{equation*}
\Tr{[\tau H_{\textnormal{PT}}]} = i\gamma v_z  + v_x s.    
\end{equation*}
where the state $\tau$ is parameterized by the vector $\vec{v} = \lbrace v_x, v_y, v_z \rbrace$.
\\
\\
Now, we will evaluate $A^{\textnormal{min}}_{\textnormal{inc}}$ considering PT-symmetric Hamiltonian as an observable over all incoherent single-qubit states
and we define it as $E^{\textnormal{min}}_{\textnormal{PT}}$. Therefore, it is mathematically given by
\begin{eqnarray}
E^{\textnormal{min}}_{\textnormal{PT}} &\coloneqq& \min_{\tau \in \mc{I}_{\textnormal{sq}}} \Tr{[\tau  H_{\textnormal{PT}} ]},\\
&=& \min_{\tau \in \mc{I}_{\textnormal{sq}}} [(p_0 - p_1)i\gamma],\\
 &=& - i \gamma,
\end{eqnarray}
where $p_0,p_1\geq 0$ and $p_0+p_1 = 1$.
It makes clear that when any single-qubit state satisfies~\eqref{PT-co-eqn}, it will be a coherent state with respect to the standard $\sigma_z$ basis. In the similar way, we can prove that a single-qubit state will be a coherent state if it satisfies~\eqref{APT-co-eqn} considering the APT-symmetric Hamiltonian as an observable.
\end{proof}
}

{
\subsection{Detection of entanglement in the ground state of the antiferromagnetic \texorpdfstring{$\textnormal{XXX}$}{XXX} model}
\label{ent-XXX}
\begin{lemma}
\label{ent-detect}
In the thermodynamic limit, the ground state of the antiferromagnetic $\textnormal{XXX}$ model is an entangled state.
\end{lemma}
\begin{proof}
Let us consider a tripartite system composed of subsystems $A_1$, $A_2$, and $A_3$ and $\mathcal{H}_{A_1} \otimes \mc{H}_{A_2} \otimes \mc{H}_{A_3}$ denotes the Hilbert space of the composite system.
The Hamiltonian ($H_{\textnormal{XXX}}^{\textnormal{tri}}$) of the antiferromagnetic tripartite XXX model with coupling constant $J$ and periodic boundary conditions, i.e., $\vec{\sigma}^{N+1} = \vec{\sigma}^N$ where $N=3$ can be written in terms of the swap operator as follows:
\begin{align}
H_{\textnormal{XXX}}^{\textnormal{tri}} &= J \sum_{l=1}^{3} \sum_{i=1}^{3} \sigma^{l}_{i} \otimes \sigma^{l+1}_{i} ,\\
&= -3J + 2J \sum_{i=1}^{3} F_{A_i A_{i+1}} \label{H-tri-XXX},
\end{align}
where $l$ denotes $l$-th site of the lattice.
$F_{A_iA_{i+1}}$ corresponds to the swap operator between the subsystems $A_i$ and $A_{i+1}$ of the tripartite system.
\\
\\
Let us suppose that $\mathscr{S}_{\mathrm{FS}}$ represents a set consisting of all the fully separable states of a tripartite system on the Hilbert space $\mathcal{H}_{A_1} \otimes \mc{H}_{A_2} \otimes \mc{H}_{A_3}$ with $\dim(\mc{H}_{A_1}) = \dim(\mc{H}_{A_2}) = \dim(\mc{H}_{A_3}) = 2$.
The minimum mean energy over the set $\mathscr{S}_{\mathrm{FS}}$ considering  the Hamiltonian of antiferromagnetic tripartite XXX model with coupling constant $J$ as a Hamiltonian is denoted as $E^{\textnormal{min}}_{\mathscr{S}_{\mathrm{FS}}}$ and it will have
\begin{align}
\label{GME-eqn-1}
E^{\textnormal{min}}_{\mathscr{S}_{\mathrm{FS}}} \nonumber  &=  \min_{\varepsilon \in \mathscr{S}_{\mathrm{FS}}} \Tr \left[ \varepsilon H_{\textnormal{XXX}}^{\textnormal{tri}} \right],\\ 
&= \min_{\varepsilon \in \mathscr{S}_{\mathrm{FS}}} \Tr \left[ \left( \sum_k q_k \rho_{A_1}^k \otimes  \rho_{A_2}^k \otimes \rho_{A_3}^k  \right)
H_{\textnormal{XXX}}^{\textnormal{tri}} \right] ,\\
&= \min_{\varepsilon \in \mathscr{S}_{\mathrm{FS}}} \left[ \sum_k q_k \left(  -3J + 2J\sum_{i=1}^{3}(1-2p_i) \right)_k 
\right], \label{GME-eqn-2}\\
&= -3J, \label{GME-eqn-3}
\end{align}
where $\sum_k q_k =1$. We have obtained Eq.~\eqref{GME-eqn-2} using Eq.~\eqref{H-tri-XXX}.
In Eq.~\eqref{GME-eqn-3}, we have used the fact that the separability range of the bipartite Werner state parameterized by $p$ on the Hilbert space with dimension $d$ is $\frac{(d-1)}{2d} \leq p \leq \frac{1}{2}$~\cite{steer-wer-2} and each $p_i
$ for $i \in \lbrace 1,2,3\rbrace$ represents a parameter of the bipartite Werner state on $\mathbb{C}^2$. 
\\
\\
Similarly, we can obtain that the minimum mean energy over all the fully separable states 
of $N$-partite system on the Hilbert space $(\mathbb{C}^2)^{\otimes N}$ with the Hamiltonian of antiferromagnetic N-partite XXX model with coupling constant $J$ will be $(-NJ)$. Since the ground state energy of the $N$-partite antiferromagnetic XXX model with coupling constant $J$ and site number $N$ in thermodynamic limit is $-1.77NJ$, 
therefore from Lemma~\ref{lemma5}, we can conclude that the ground state of the antiferromagnetic XXX model is an entangled state in the thermodynamic limit.
\end{proof}
}

\bibliography{qm}

\end{document}